\documentclass[lettersize,journal]{IEEEtran}
\IEEEoverridecommandlockouts
\usepackage{setspace}                   
\usepackage{lineno}
\usepackage{hyperref}
\usepackage{cite}
\usepackage{amsmath,amssymb,amsfonts}
\usepackage{amsmath}
\usepackage{amsthm}

\usepackage{graphicx}
\usepackage{textcomp}
\usepackage{xcolor}
\usepackage{graphicx}
\usepackage{float}
\usepackage{subfigure}
\usepackage{amsmath,mleftright}
\usepackage{amsfonts,amssymb}
\usepackage{mathrsfs}
\usepackage{mathtools}
\usepackage{algorithm}
\usepackage{algorithmicx}
\usepackage{algpseudocode}
\usepackage{bm}
\usepackage{multirow}
\usepackage{array}
\usepackage{amssymb}
\usepackage{amsmath}
\usepackage{url}
\usepackage{xcolor}
\usepackage{graphicx,amsmath,amssymb}
\usepackage{subfigure}
\usepackage{fancyhdr}
\usepackage{mdwmath}
\usepackage{mdwtab}
\usepackage{caption}
\usepackage{amsthm}
\usepackage{setspace}
\usepackage{bm}
\usepackage{algorithm}
\usepackage{algpseudocode}
\usepackage{mathtools}
\usepackage{dsfont}
\usepackage{bbm}
\newtheorem{remark}{Remark}
\newtheorem{theorem}{Theorem}

\newtheorem{lemma}{Lemma}

\newtheorem{corollary}{Corollary}

\makeatletter
\newcommand{\biggg}{\bBigg@{3}}
\newcommand{\Biggg}{\bBigg@{3.5}}
\makeatother
\def\BibTeX{{\rm B\kern-.05em{\sc i\kern-.025em b}\kern-.08em
    T\kern-.1667em\lower.7ex\hbox{E}\kern-.125emX}}

\begin{document}
\title{Multiuser Beamforming for Pinching-Antenna Systems: An Element-wise Optimization Framework}
\author{
\author{
Mingjun Sun,~\IEEEmembership{Graduate Student Member,~IEEE,} Chongjun Ouyang,~\IEEEmembership{Member,~IEEE,}\\ Shaochuan Wu,~\IEEEmembership{Senior Member,~IEEE,} and Yuanwei Liu,~\IEEEmembership{Fellow,~IEEE}
\thanks{Mingjun Sun and Shaochuan Wu are with the School of Electronics and Information Engineering, Harbin Institute of Technology, Harbin 150001, China (e-mail: sunmj@stu.hit.edu.cn; scwu@hit.edu.cn).}
\thanks{Chongjun Ouyang is with the School of Electronic Engineering and Computer Science, Queen Mary University of London, London, E1 4NS, U.K. (e-mail: c.ouyang@qmul.ac.uk).}
\thanks{Yuanwei Liu is with the Department of Electrical and Electronic Engineering, The University of Hong Kong, Hong Kong (e-mail: yuanwei@hku.hk).}
}}
\maketitle
\begin{abstract}
The pinching-antenna system (PASS) reconstructs wireless channels through pinching beamforming, i.e., optimizing the activated locations of pinching antennas (PAs) along the waveguide. The aim of this article is to investigate the joint design of baseband beamforming and pinching beamforming. A low-complexity element-wise sequential optimization framework is proposed to address the sum-rate maximization problem in PASS-enabled downlink and uplink channels. i) For the downlink scenario, maximum ratio transmission (MRT), zero-forcing (ZF), and minimum mean square error (MMSE) beamforming schemes are employed as baseband beamformers. For each beamformer, a closed-form expression for the downlink sum-rate is derived as a single-variable function with respect to the pinching beamformer. Based on this, a sequential optimization method is proposed, where the positions of the PAs are updated element-wise using a low-complexity one-dimensional search. ii) For the uplink scenario, signal detection is performed using maximum ratio combining (MRC), ZF, and MMSE combiners. A closed-form sum-rate expression is derived for each linear combiner, and a similar element-wise design is applied to optimize the pinching beamforming.
Numerical results are provided to validate the effectiveness of the proposed method and demonstrate that: (i) For all considered linear beamformers, the proposed PASS architecture outperforms conventional fixed-antenna systems in terms of sum-rate performance; (ii) in both downlink and uplink channels, ZF achieves performance close to that of MMSE and significantly outperforms MRT or MRC; and (iii) the proposed element-wise design eliminates the need for alternating updates between the baseband and pinching beamformers, thereby ensuring low computational complexity.

\end{abstract} 
\section{Introduction}
As a cornerstone of wireless communications for decades, multiple-input multiple-output (MIMO) technology continues to maintain its pivotal status in modern systems\cite{heath2018foundations, MIMO1}. In recent years, MIMO has been evolving rapidly to address the explosive growth in demand expected in future 6G networks. As part of this evolution, a new class of flexible-antenna systems has emerged, including reconfigurable intelligent surfaces (RIS)\cite{RIS, Wu2020}, fluid antennas\cite{wong2020fluid}, and movable antennas\cite{zhu2023movable}. The core idea of them lies in introducing new degrees of freedom (DoFs) to enable reconfigurable or even customized wireless channels. This allows for tunable effective channel gains tailored to served users, thereby enhancing system throughput. Such technologies have garnered significant research interest in both academia and industry.

Nevertheless, these technologies are more effective in mitigating small-scale fading, while their capability to combat large-scale path loss remains limited due to their small spatial tuning range of antenna elements. As a result, they are generally unable to address line-of-sight (LoS) path blockage issues. Although RISs can enable signal diffraction around obstacles by programming the phase shifts of its metasurface elements, this comes at the cost of double attenuation\cite{double_attenuation}. Moreover, these architectures often lack hardware flexibility, as the antenna arrays are typically fixed once fabricated, making it difficult to scale the number of antennas dynamically. To address these limitations, a novel flexible-antenna architecture known as the pinching-antenna system (PASS) has been proposed\cite{liu2025pinchingantenna}.

The first PASS prototype, developed by DOCOMO\cite{suzuki2022pinching, pinching_antenna1}, comprises elongated dielectric waveguides and small dielectric particles, referred to as pinching antennas (PAs). The dielectric waveguides enable near-lossless signal transmission, while the PAs induce electromagnetic wave emission and reception. Owing to the adjustable positions of PAs, PASS enables a hybrid beamforming architecture. The transmit signal is first processed by a baseband beamformer and then fed into the waveguides, where the PAs' positions are further adjusted to tune the signal phase and control large-scale path loss, which is referred to as \emph{pinching beamforming}.
This physical implementation of PASS brings the following advantages: 1) \emph{Strong LoS link}: Compared to the aforementioned flexible antenna schemes, PASS can extend the dielectric waveguide arbitrarily. This allows the deployment of PAs closer to users to form strong LoS links. As a result, large-scale path loss is reduced, which significantly improves system performance. 2) \emph{Scalable and low-cost deployment}: PASS enables flexible control of the formation or termination of communication regions by adding or removing PAs. This operation is simple and low-cost in terms of hardware, which is also not achievable by existing flexible-antenna systems.

In essence, PASS can be viewed as a practical realization of the fluid-antenna and movable-antenna concepts proposed in prior works \cite{fluidante}, \cite{movableante}, while offering enhanced flexibility and scalability. In recognition of DOCOMO's foundational contributions, we refer to this technology as PASS throughout this article. Moreover, PASS aligns with the emerging vision of surface-wave communication superhighways \cite{Wong}, which envisions leveraging in-waveguide propagation through reconfigurable waveguides to reduce path loss and improve signal power delivery \cite{LiuHaizhe}, \cite{ChuZhiyuan}.

\subsection{Prior Works}
Driven by the benefits of PASS, several recent studies have focused on exploring its potential, including performance analysis\cite{ding2024flexible, tianwei2025ULperformance, tyrovolas2025performanceanalysis, ouyang2025array} and beamforming design\cite{tegos2024minimum, YanqingMaximization, NOMA, bereyhi2025downlink, bereyhi2025mimopass,wang2025modeling}. 
For performance analysis, the authors of \cite{ding2024flexible} investigated various downlink PASS configurations and demonstrated the significant performance gains of PASS over conventional fixed-antenna systems using stochastic geometry. The work in \cite{tianwei2025ULperformance} further derived the ergodic uplink rate for PASS. In \cite{tyrovolas2025performanceanalysis}, the authors examined key performance metrics, such as outage probability and data rate, under lossy waveguide conditions, revealing their critical impact, especially in the case of long waveguides. In addition, the optimal number of antennas and inter-antenna spacing in PASS was studied in \cite{ouyang2025array}.

The hybrid beamforming architecture necessitates the joint optimization of baseband and pinching beamforming, which poses significant challenges. To address the uplink user fairness problem, \cite{tegos2024minimum} decomposed the joint optimization into two subproblems: the PAs' locations were optimized via successive convex approximation (SCA), and a closed-form solution for power allocation was subsequently derived. For downlink rate maximization, a two-stage algorithm was proposed in \cite{YanqingMaximization} for a single-waveguide, single-user scenario, where the PAs' positions were first coarsely optimized to minimize path loss, then refined to enhance signal strength. In \cite{NOMA}, a nonorthogonal multiple access (NOMA)-assisted multiuser communication system was considered, where the authors proposed a matching-theory-based single-waveguide antenna activation method to maximize the system throughput.
\cite{bereyhi2025downlink} considered a multi-waveguide, multi-user scenario where each waveguide activated a single PA. The authors proposed a block coordinate descent (BCD) algorithm based on fractional programming (FP) and Gauss-Seidel methods, which alternately optimized baseband beamforming and PAs' positions to maximize the weighted sum-rate. This work was extended in \cite{bereyhi2025mimopass} to a more general case where each waveguide activated multiple PAs. 
In addition, the downlink transmit power minimization problem was addressed in \cite{wang2025modeling}, where a penalty-based alternating optimization algorithm was proposed to minimize transmit power.
The above studies have collectively demonstrated the effectiveness of PASS in enhancing the spectral efficiency and its superiority over traditional fixed-antenna systems. 

\subsection{Motivations and Contributions}
For sum-rate maximization problem in PASS, the coupling between baseband and pinching beamforming typically necessitates alternating optimization, as adopted in existing methods \cite{bereyhi2025downlink, bereyhi2025mimopass, wang2025modeling}, all involving nested inner and outer iterations. This incurs significant computational overhead, particularly in large-user scenarios with high-dimensional baseband beamforming caused by the large number of radio-frequency (RF) chains.
Moreover, the joint optimization algorithms in \cite{bereyhi2025mimopass, wang2025modeling} are sensitive to the initialization of the PAs' positions, and their performance is generally comparable to that of the low-complexity zero-forcing (ZF)-based scheme. 
Motivated by this observation, we adopt classical linear precoding schemes as baseband beamformers and derive a closed-form, structurally simplified update rule for PAs' positions under each baseband beamforming strategy.
Unlike existing methods, the proposed approach eliminates the need for iterative alternating optimization, thereby significantly reducing computational complexity.


The main contributions of this paper are summarized as follows:
\begin{itemize}
    \item To evaluate the performance gain of PASS over conventional fixed-antenna systems, we formulate sum-rate maximization problems for both downlink and uplink transmissions under a practical model that incorporates in-waveguide loss. To solve these problems, we propose an element-wise sequential optimization framework that enables joint optimization of baseband beamforming and pinching beamforming.
	\item For downlink transmission, we consider three baseband beamforming schemes: maximum ratio transmission (MRT), ZF, and minimum mean square error (MMSE). We prove by contradiction that the inequality power constraint is equivalent to an equality constraint, which allows power normalization of the baseband beamformers. Then, we derive simplified closed-form expressions for the objective functions with respect to the PAs' positions. Finally, a low-complexity one-dimensional search algorithm is applied to sequentially optimize the PAs' positions, maximizing system throughput with significantly reduced complexity.
	\item For uplink signal detection, three pinching beamforming schemes based on maximum ratio combining (MRC), ZF, and MMSE are investigated. Similar to the downlink case, we apply the matrix inversion lemma to derive a simplified PAs' positions update problem. This design significantly reduces reliance on high-complexity operations such as matrix inversion and determinant calculation, thereby improving computational efficiency. During each sequential update of the PAs' positions, the baseband beamformer is simultaneously adjusted to remain aligned with the channel.

	\item We present numerical results to validate the effectiveness of the proposed method and demonstrate the significant potential of PASS compared to conventional fixed-antenna systems.
Specifically:
1) By avoiding the ``inner-outer double-loop'' iteration structure, the proposed method yields a low computational overhead.
2) Across all considered system configurations, PASS consistently outperforms conventional fixed-antenna systems.
3) Compared to conventional fixed-antenna systems, the signal-to-noise ratio (SNR) reference thresholds in PASS that separate the high- and low-SNR performance regions for MRT (MRC) and ZF are lower, and the performance gap between ZF and MMSE becomes negligible.
\end{itemize}

The remainder of this paper is organized as follows. Section II formulates the sum-rate maximization problems for both downlink and uplink transmissions. The proposed pinching beamforming methods for downlink and uplink are derived in Sections III and IV, respectively. Numerical results are provided in Section V. Finally, conclusions are drawn in Section VI.

\textit{Notation:} 
Unless otherwise specified, $\mathbb{R}$ and $\mathbb{C}$ denote the real and complex number sets, respectively. The $N \times N$ identity matrix is denoted by $\mathbf{I}_N$, and $[N]$ represents the set $\{1, \ldots, N\}$. Vectors and matrices are represented by bold lower-case and upper-case letters, respectively. For any matrix $\mathbf{A}$, $\mathbf{A}^{\mathrm{T}}$, $\mathbf{A}^{*}$, and $\mathbf{A}^{\mathrm{H}}$ denote its transpose, conjugate, and conjugate-transpose, respectively.

\section{System Model}\label{Section: System Model}
Consider a PASS, where a base station (BS) serves $K$ single-antenna users using $M$ pinched waveguides, as shown in {\figurename} {\ref{Figure: System_Model}}. The users are distributed within a predefined two-dimensional rectangular area of size $D_x\times D_y$. We denote the location of user $k\in[K]$ by ${\mathbf{u}}_k=[x_k, y_k, 0]^{\mathsf{T}}$, assuming that all users are situated on the $xy$-plane.

\subsection{Signal Model}
Referring to {\figurename} {\ref{Figure: System_Model}}, the $M$ waveguides are deployed along the $x$-axis at a height $a$ and are arranged in an array along the $y$-axis, with adjacent waveguides spaced by $d$. We assume that each waveguide is equipped with $N$ PAs. The location of the $n$th PA on the $m$th waveguide (referred to as the $(m,n)$th PA) is denoted by ${\mathbf{p}}_{m,n}=[p_{m,n},(m-1)d,a]^{\mathsf{T}}$, where $p_{m,n}$ represents the position of the $(m,n)$th PA along the $x$-axis. Additionally, the distance between any two adjacent pre-configurable positions, $\lvert p_{m,n}-p_{m,n'}\rvert$ ($\forall n\ne n'$), should be greater than or equal to half the wavelength of the carrier frequency to avoid mutual coupling effects.

\begin{figure}[!t]
 \centering
\setlength{\abovecaptionskip}{0pt}
\includegraphics[height=0.2\textwidth]{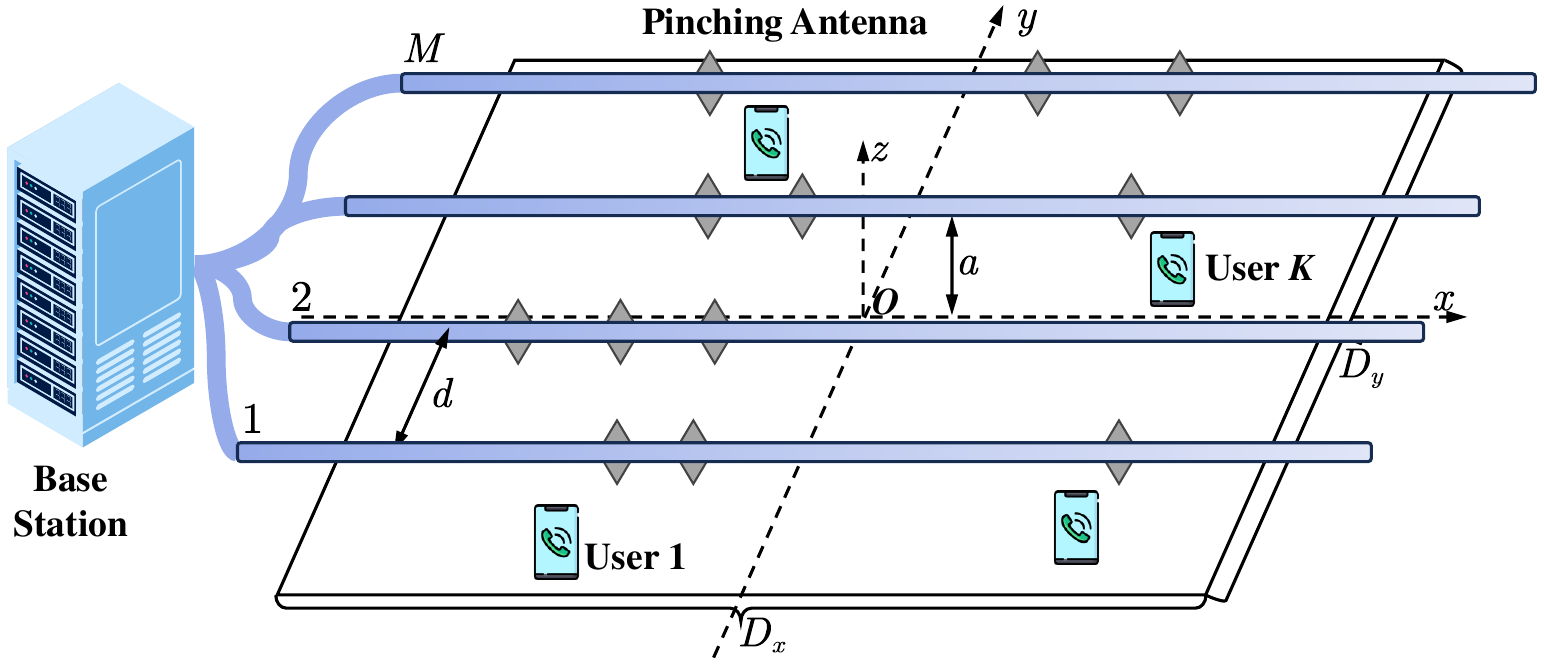}
\caption{Illustration of a PASS where the BS is equipped by multiple pinched waveguides.}
\label{Figure: System_Model}
\vspace{-15pt}
\end{figure}

According to \cite{pinching_antenna1,suzuki2022pinching}, a dielectric waveguide can transmit and receive radio waves or signals through its pinched dielectric particles. For simplicity, we model each PA on the waveguide as an isotropic antenna that feeds the signal into or out of the waveguide. 
\subsubsection{Downlink Model}
For downlink transmission, the signal received by user $k$ from all pinching antennas can be expressed as follows:
\begin{align}\label{DL_Signal_Model_Initial}
y_k=\sum_{m=1}^{M}\sum_{n=1}^{N}h_{k,m,n}x_{m,n}+z_k,
\end{align}
where $x_{m,n}$ represents the signal radiated from the $(m,n)$th PA, $h_{k,m,n}$ denotes the channel coefficient between user $k$ and the $(m,n)$th PA, $z_k\sim{\mathcal{CN}}(0,\sigma_k^2)$ is the additive white Gaussian noise (AWGN) at user $k$, with $\sigma_k^2$ representing the noise power. Taking into account the wave propagation from the BS to the PA along the waveguide, we have
\begin{align}\label{Inner_Waveguide_Signal}
x_{m,n}=\eta_{m,n}^{\frac{1}{2}}{\rm{e}}^{-{\rm{j}}k_{\rm{g}}(p_{m,n}-o_m)}x_{m},
\end{align}
where $k_{\rm{g}}=\frac{2\pi}{\lambda_{\rm{g}}}$ is the wavenumber in waveguides with $\lambda_{\rm{g}}$ representing the waveguide wavelength in a dielectric medium, $o_m$ is the position on the $m$th waveguide where the signal is fed from the BS, and $\eta_{m,n}$ accounts for the power loss along the dielectric waveguide. Based on \cite{NOMA}, we model $\eta_{m,n}$ as follows:
\begin{align}
\eta_{m,n}=\frac{10^{-\frac{\kappa\left\lvert p_{m,n}-o_m\right\rvert}{10}}}{N},
\end{align}
where $\kappa\geq0$ is the average attenuation factor along the dielectric waveguide in dB/m \cite{yeh2008essence}. We note that $\kappa=0$ represents the special case of a lossless dielectric and a perfectly conducting surface. 

Substituting \eqref{Inner_Waveguide_Signal} into \eqref{DL_Signal_Model_Initial} gives 
\begin{align}\label{DL_Signal_Model_Second}
y_k=\sum_{m=1}^{M}\sum_{n=1}^{N}h_{k,m,n}\eta_{m,n}^{\frac{1}{2}}{\rm{e}}^{-{\rm{j}}k_{\rm{g}}(p_{m,n}-o_m)}x_{m}+z_k.
\end{align}
For compactness, we rewrite \eqref{DL_Signal_Model_Second} as follows:
\begin{align}\label{DL_Signal_Model}
y_k={\mathbf{h}}_k^{\mathsf{H}}({\mathbf{P}}){\mathbf{G}}_{\rm{dl}}({\mathbf{P}}){\mathbf{x}}+z_k.
\end{align}
The terms appearing in \eqref{DL_Signal_Model} are defined as follows:
\begin{itemize}
  \item ${\mathbf{P}}=[{\mathbf{p}}_{1},\ldots,{\mathbf{p}}_{M}]\in{\mathbbmss{R}}^{N\times M}$ consists of the $x$-axis positions of all PAs, where ${\mathbf{p}}_{m}=[p_{m,1},\ldots,p_{m,N}]^{\mathsf{T}}\in{\mathbbmss{R}}^{N\times 1}$ stores the positions of the PAs along the $m$th waveguide for $m\in[M]$.
  \item ${\mathbf{h}}_k({\mathbf{P}})=[\bar{\mathbf{h}}_k^{\mathsf{T}}({\mathbf{p}}_{1}),\ldots,\bar{\mathbf{h}}_k^{\mathsf{T}}({\mathbf{p}}_{M})]^{\mathsf{H}}\in{\mathbbmss{C}}^{MN\times1}$ denotes the channel vector for user $k$, where $\bar{\mathbf{h}}_k({\mathbf{p}}_{m})=[h_{k,m,1},\ldots,h_{k,m,N}]^{\mathsf{T}}\in{\mathbbmss{C}}^{N\times 1}$ denotes the channel vector from user $k$ to the $m$th waveguide for $m\in[M]$.
  \item ${\mathbf{G}}_{\rm{dl}}({\mathbf{P}})={\rm{BlkDiag}}({\mathbf{g}}_{\rm{dl}}({\mathbf{p}}_{1});\ldots;{\mathbf{g}}_{\rm{dl}}({\mathbf{p}}_{M}))\in{\mathbbmss{C}}^{MN\times M}$ accounts for wave propagation effects within the waveguides, where ${\mathbf{g}}_{\rm{dl}}({\mathbf{p}}_{m})=[\eta_{m,1}^{\frac{1}{2}}{\rm{e}}^{-{\rm{j}}k_{\rm{g}}(p_{m,1}-o_m)},\ldots
      ,\eta_{m,N}^{\frac{1}{2}}{\rm{e}}^{-{\rm{j}}k_{\rm{g}}(p_{m,N}-o_m)}]^{\mathsf{T}}\in{\mathbbmss{C}}^{N\times 1}$ for $m\in[M]$. 
  \item ${\mathbf{x}}=[x_1,\ldots,x_M]^{\mathsf{T}}={\mathbf{W}}{\mathbf{s}}=\sum_{k=1}^{K}{\mathbf{w}}_ks_k\in{\mathbbmss{C}}^{M\times1}$ denotes the signal after transmit beamforming, where ${\mathbf{W}}=[{\mathbf{w}}_1,\ldots,{\mathbf{w}}_K]\in{\mathbbmss{C}}^{M\times K}$ denotes the transmit beamforming matrix with ${\mathbf{w}}_k\in{\mathbbmss{C}}^{M\times1}$ being the beamforming vector for user $k$, and ${\mathbf{s}}=[s_1,\ldots,s_K]^{\mathsf{T}}\in{\mathbbmss{C}}^{K\times 1}$ denotes the coded symbol vector with ${\mathbf{s}}_k$ being the data symbol for user $k$ and ${\mathbf{s}}\sim{\mathcal{CN}}({\mathbf{0}},{\mathbf{I}}_K)$.
\end{itemize}
It follows that
\begin{align}
y_k=\underbrace{{\mathbf{h}}_k^{\mathsf{H}}({\mathbf{P}}){\mathbf{G}}_{\rm{dl}}({\mathbf{P}}){\mathbf{w}}_ks_k}_{\text{desired~signal}}
+\underbrace{{\mathbf{h}}_k^{\mathsf{H}}({\mathbf{P}}){\mathbf{G}}_{\rm{dl}}({\mathbf{P}})\sum_{k'\ne k}{\mathbf{w}}_{k'}s_{k'}}_{\text{inter-user~interference}}+z_k.
\end{align}
The received signal-to-interference-plus-noise ratio (SINR) at user $k$ to decode the data information is thus given as follows:
\begin{align}\label{SINR_dl}
\gamma_{{\rm{dl}},k}=\frac{\lvert{\mathbf{h}}_k^{\mathsf{H}}({\mathbf{P}}){\mathbf{G}}_{\rm{dl}}({\mathbf{P}}){\mathbf{w}}_k\rvert^2}
{\sum_{k'\ne k}\lvert{\mathbf{h}}_k^{\mathsf{H}}({\mathbf{P}}){\mathbf{G}}_{\rm{dl}}({\mathbf{P}}){\mathbf{w}}_{k'}\rvert^2+\sigma_k^2}.
\end{align}
\subsubsection{Uplink Model}
For uplink transmission operating in the time-domain duplexing (TDD) mode, the signal received at the $(m,n)$th PA can be expressed as follows:
\begin{align}
\hat{y}_{m,n}=\sum_{k=1}^{K}h_{k,m,n}x_k+z_{m,n},
\end{align}
where $x_k$ is the signal transmitted by user $k$, and $z_{m,n}\sim{\mathcal{CN}}(0,\sigma^2)$ is the AWGN at the PA, with $\sigma^2$ representing the noise power. Given the described system, we express the transmitted signal by user $k$ as $x_k=\sqrt{P_k}s_k$, where $P_k$ is the transmit power, and $s_k$ is the normalized encoded data symbol satisfying ${\mathbbmss{E}}\{\lvert s_k\rvert^2\}=1$ for $k\in[K]$. Let ${\mathbf{s}}=[s_1,\ldots,s_K]^{\mathsf{T}}$ denote the vector of all users' data symbols, which are assumed to follow a zero-mean unit-variance complex Gaussian distribution, i.e., ${\mathbf{s}}\sim{\mathcal{CN}}({\mathbf{0}},{\mathbf{I}}_K)$.

Taking into account the wave propagation from the PA to the BS along the waveguide, the signal received at the BS from the $(m,n)$th PA is given by
\begin{align}\label{Received_Signal_Each_PA_to_BS}
y_{m,n}=\zeta_{m,n}^{\frac{1}{2}}{\rm{e}}^{-{\rm{j}}k_{\rm{g}}(p_{m,n}-o_m)}\hat{y}_{m,n},
\end{align}
where $\zeta_{m,n}=10^{-\frac{\kappa\left\lvert p_{m,n}-o_m\right\rvert}{10}}$ accounts for the power loss along the dielectric waveguide. 

Building on \eqref{Received_Signal_Each_PA_to_BS}, we write the total received signal vector at the BS as follows:
\begin{align}\label{Received_Signal_at_BS}
{\mathbf{y}}=\left[\begin{smallmatrix}\sum_{n=1}^{N}y_{1,n}\\\vdots\\\sum_{n=1}^{N}y_{M,n}\end{smallmatrix}\right]
={\mathbf{G}}_{\rm{ul}}^{\mathsf{T}}({\mathbf{P}})\sum_{k=1}^{K}{\mathbf{h}}_k^{*}({\mathbf{P}})\sqrt{P_k}s_k+{\mathbf{G}}_{\rm{ul}}^{\mathsf{T}}({\mathbf{P}}){\mathbf{z}}.
\end{align}
The terms appearing in \eqref{Received_Signal_at_BS} are defined as follows:
\begin{itemize}
  \item ${\mathbf{z}}=[[z_{1,1},\ldots,z_{1,N}],\ldots,[z_{M,1},\ldots,z_{M,N}]]^{\mathsf{T}}\in{\mathbbmss{C}}^{MN\times1}$ is the AWGN vector satisfying ${\mathbf{z}}\sim{\mathcal{CN}}({\mathbf{0}},\sigma^2{\mathbf{I}}_{MN})$.
  \item ${\mathbf{G}}_{\rm{ul}}({\mathbf{P}})={\rm{BlkDiag}}({\mathbf{g}}_{\rm{ul}}({\mathbf{p}}_{1});\ldots;{\mathbf{g}}_{\rm{ul}}({\mathbf{p}}_{M}))\in{\mathbbmss{C}}^{MN\times M}$ accounts for wave propagation effects within the waveguides, where ${\mathbf{g}}_{\rm{ul}}({\mathbf{p}}_{m})=[\zeta_{m,1}^{\frac{1}{2}}{\rm{e}}^{-{\rm{j}}k_{\rm{g}}(p_{m,1}-o_m)},\ldots
      ,\zeta_{m,N}^{\frac{1}{2}}{\rm{e}}^{-{\rm{j}}k_{\rm{g}}(p_{m,N}-o_m)}]^{\mathsf{T}}\in{\mathbbmss{C}}^{N\times 1}$ for $m\in[M]$.
\end{itemize}
It follows that ${\mathbf{G}}_{\rm{ul}}^{\mathsf{T}}({\mathbf{P}}){\mathbf{z}}$ is Gaussian distributed with its mean and covariance determined as follows:
\begin{subequations}
\begin{align}
&{\mathbbmss{E}}\{{\mathbf{G}}_{\rm{ul}}^{\mathsf{T}}({\mathbf{P}}){\mathbf{z}}\}={\mathbf{G}}_{\rm{ul}}^{\mathsf{T}}({\mathbf{P}}){\mathbbmss{E}}\{{\mathbf{z}}\}={\mathbf{0}},\\
&{\mathbbmss{E}}\{{\mathbf{G}}_{\rm{ul}}^{\mathsf{T}}({\mathbf{P}}){\mathbf{z}}({\mathbf{G}}_{\rm{ul}}^{\mathsf{T}}({\mathbf{P}}){\mathbf{z}})^{\mathsf{H}}\}\\\nonumber
&={\rm{Diag}}([\sigma^2\lVert{\mathbf{g}}_{\rm{ul}}({\mathbf{p}}_{1})\rVert^2,\ldots,\sigma^2\lVert{\mathbf{g}}_{\rm{ul}}({\mathbf{p}}_{M})\rVert^2])\triangleq{\mathbf{R}}_{\mathbf{z}}({\mathbf{P}}),
\end{align}
\end{subequations}
where $\lVert{\mathbf{g}}_{\rm{ul}}({\mathbf{p}}_{m})\rVert^2=\sum_{n=1}^{N}\zeta_{m,n}$ for $m\in[M]$. Therefore, we obtain ${\mathbf{G}}_{\rm{ul}}^{\mathsf{T}}({\mathbf{P}}){\mathbf{z}}\sim{\mathcal{CN}}({\mathbf{0}},{\mathbf{R}}_{\mathbf{z}}({\mathbf{P}}))$. Let $\mathbf{v}_k\in{\mathbbmss{C}}^{M\times1}$ denote the linear detector used at the BS to recover the data information contained in $s_k$, which yields
\begin{align}
{\mathbf{v}}_k^{\mathsf{H}}{\mathbf{y}}&=\underbrace{{\mathbf{v}}_k^{\mathsf{H}}{\mathbf{G}}_{\rm{ul}}^{\mathsf{T}}({\mathbf{P}}){\mathbf{h}}_k^{*}({\mathbf{P}})
\sqrt{P_k}s_k}_{\text{desired~signal}}\nonumber \\ 
&+\underbrace{{\mathbf{v}}_k^{\mathsf{H}}{\mathbf{G}}_{\rm{ul}}^{\mathsf{T}}({\mathbf{P}})\sum_{k'\ne k}{\mathbf{h}}_{k'}^{*}({\mathbf{P}})\sqrt{P_{k'}}s_{k'}}_{\text{inter-user~interference}}
+{\mathbf{v}}_k^{\mathsf{H}}{\mathbf{G}}_{\rm{ul}}^{\mathsf{T}}({\mathbf{P}}){\mathbf{z}},
\end{align}
with ${\mathbf{v}}_k^{\mathsf{H}}{\mathbf{G}}_{\rm{ul}}^{\mathsf{T}}({\mathbf{P}}){\mathbf{z}}\sim{\mathcal{CN}}(0,{\mathbf{v}}_k^{\mathsf{H}}{\mathbf{R}}_{\mathbf{z}}({\mathbf{P}}){\mathbf{v}}_k)$ denoting the effective noise. Therefore, the resulting SINR for decoding $s_k$ can be given as follows:
\begin{align}\label{per_user_sinr}
\gamma_{{\rm{ul}},k}=\frac{P_k\lvert{\mathbf{v}}_k^{\mathsf{H}}{\mathbf{G}}_{\rm{ul}}^{\mathsf{T}}({\mathbf{P}}){\mathbf{h}}_k^{*}({\mathbf{P}})\rvert^2}
{\sum_{k'\ne k}P_{k'}\lvert{\mathbf{v}}_k^{\mathsf{H}}{\mathbf{G}}_{\rm{ul}}^{\mathsf{T}}({\mathbf{P}}){\mathbf{h}}_{k'}^{*}({\mathbf{P}})\rvert^2+{\mathbf{v}}_k^{\mathsf{H}}{\mathbf{R}}_{\mathbf{z}}({\mathbf{P}}){\mathbf{v}}_k}.
\end{align}
\subsection{Channel Model}
Similar to existing works \cite{ding2024flexible,tegos2024minimum,ouyang2025array,wang2025modeling}, we assume that the users have LoS links to the waveguides, which is typically the case in indoor environments. Non-LoS (NLoS) paths are considered negligible due to their significantly weaker signal strength compared to LoS paths. Since each PA acts as an isotropic receiver, the channel coefficient between user $k$ and the $(m,n)$th PA is expressed as follows \cite{liu2023near-field}:
\begin{align}
h_{k,m,n}=\frac{\lambda{\rm{e}}^{-{\rm{j}}k_0 d_k(p_{m,n})}}{4\pi d_k(p_{m,n})},
\end{align}
where $d_k(p_{m,n}) = \lVert{\mathbf{p}}_{m,n}-{\mathbf{u}}_k\rVert$ denotes the Euclidean distance between the $(m,n)$th PA and user $k$, $\lambda$ is the carrier wavelength in free space, and $k_0=\frac{2\pi}{\lambda}$ is the wavenumber.
\subsection{Problem Formulation}
In this paper, we aim to maximize the downlink and uplink sum-rates by jointly optimizing the \emph{transmit and receive beamforming} at the BS and the \emph{pinching beamforming} facilitated by the PAs.
\subsubsection{Downlink PASS}
Denote ${\mathbf{W}}=[{\mathbf{w}}_1,\ldots,{\mathbf{w}}_K]\in{\mathbbmss{C}}^{M\times K}$. The problem of maximizing the downlink sum-rate can be formulated as follows:
\begin{subequations}\label{DL_Sum_Rate_Problem}
\begin{align}
\max_{{\mathbf{W}},{\mathbf{P}}}&~{\mathcal{R}}_{\rm{dl}}({\mathbf{W}},{\mathbf{P}})=\sum_{k=1}^{K}\log_2(1+\gamma_{{\rm{dl}},k})\\
{\rm{s.t.}}&~{\rm{tr}}({\mathbf{W}}{\mathbf{W}}^{\mathsf{H}})=\sum_{k=1}^{K}\lVert{\mathbf{w}}_k\rVert^2\leq P,\\
&~ p_{m,n}-p_{m,n-1}\geq \Delta,\forall m\in[M],n\ne1,\label{dl_c1}\\
&~p_{m,n}-o_m\in[0,L_m],\forall m\in[M],n\in[N],\label{dl_c2}
\end{align}
\end{subequations}
where $P$ is the power budget, $\Delta>0$ is the minimum spacing required to prevent mutual coupling between the PAs, and $L_m>0$ is the maximum available length of the $m$th waveguide for $m\in[M]$.
\subsubsection{Uplink PASS}
Let ${\mathbf{V}}=[{\mathbf{v}}_1,\ldots,{\mathbf{v}}_K]\in{\mathbbmss{C}}^{M\times K}$. The problem of maximizing the uplink sum-rate can be formulated as follows:
\begin{subequations}\label{UL_Sum_Rate_Problem}
\begin{align}
\max_{{\mathbf{V}},{\mathbf{P}}}&~{\mathcal{R}}_{\rm{ul}}({\mathbf{V}},{\mathbf{P}})=\sum_{k=1}^{K}\log_2(1+\gamma_{{\rm{ul}},k})\\
{\rm{s.t.}}&~ p_{m,n}-p_{m,n-1}\geq \Delta,\forall m\in[M],n\ne1,\label{ul_c1}\\
&~p_{m,n}-o_m\in[0,L_m],\forall m\in[M],n\in[N].\label{ul_c2}
\end{align}
\end{subequations}

Constraints \eqref{dl_c1} and \eqref{ul_c1} consider a practical physical condition where each PA can slide along its corresponding waveguide, similar to an electrical curtain. The sliding range of each PA is not only limited by the waveguide length but also constrained by the positions of adjacent PAs. However, according to \cite[\textit{Lemma~1}]{bereyhi2025mimopass}, the system throughput is invariant to the permutation of PAs. Therefore, these constraints can be simplified as $|p_{m,n}-p_{m,n'}|\geq \Delta,\forall m\in[M],n\ne n'$. For simplicity, in the following discussions, the constraints referred to as \eqref{dl_c1} and \eqref{ul_c1} will both denote this relaxed form.

\section{Downlink Beamforming Design}
Jointly optimizing ${\mathbf{W}}$ and ${\mathbf{P}}$ is a challenging task, and solving such a problem is generally time-consuming. As a compromise, in this work, we consider that the transmit beamforming is designed using classical linear approaches, including MRT, ZF, and MMSE. For convenience, we denote 
\begin{align}
{\mathbf{H}}_{\rm{dl}}({\mathbf{P}})=\left[\begin{matrix}{\mathbf{h}}_1^{\mathsf{H}}({\mathbf{P}}){\mathbf{G}}_{\rm{dl}}({\mathbf{P}})\\\vdots\\
{\mathbf{h}}_K^{\mathsf{H}}({\mathbf{P}}){\mathbf{G}}_{\rm{dl}}({\mathbf{P}})\end{matrix}\right]
=\left[\begin{matrix}{\mathbf{h}}_1^{\mathsf{H}}({\mathbf{P}})\\\vdots\\
{\mathbf{h}}_K^{\mathsf{H}}({\mathbf{P}})\end{matrix}\right]{\mathbf{G}}_{\rm{dl}}({\mathbf{P}})\in{\mathbbmss{C}}^{K\times M}
\end{align}
as the effective channel.

Recalling problem \eqref{DL_Sum_Rate_Problem}, the following lemma shows that the power inequality constraint can be equivalently reformulated as an equality constraint. Leveraging this lemma, we normalize the transmit power for subsequent linear beamforming methods. Based on this, we then derive the update rules for the positions of the PAs corresponding to each scheme.
\begin{lemma}\label{Equality power constraint}
The optimal solution $\mathbf{W}^{\star}$ to problem \eqref{DL_Sum_Rate_Problem} always satisfies the constraint ${\rm{tr}}(\mathbf{W}^{\star}{\mathbf{W}^{\star}}^{\mathsf{H}})=\sum_{k=1}^{K}\lVert{\mathbf{w}}_k^{\star}\rVert^2= P$.
\end{lemma}
\begin{proof}
This lemma can be proven by contradiction. Suppose $\{\hat{\mathbf{w}}_k\}_{k=1}^K$ is the solution to problem \eqref{DL_Sum_Rate_Problem}, satisfying the power constraint $\sum_{k=1}^{K}\|\hat{\mathbf{w}}_k\|^2 = \hat{P} < P$. Then, we define a scaled solution $\mathbf{w}_k = \sqrt{\frac{P}{\hat{P}}}\hat{\mathbf{w}}_k$, which satisfies $\sum_{k=1}^{K}\|\mathbf{w}_k\|^2 = P$. Substituting $\hat{\mathbf{w}}_k$ and $\mathbf{w}_k$ into \eqref{SINR_dl} yields $\hat{\gamma}_{{\rm{dl}},k}$ and $\gamma_{{\rm{dl}},k}$, respectively. 
It can be readily shown that $\hat{\gamma}_{{\rm{dl}},k}< \gamma_{{\rm{dl}},k}$. This contradicts the initial assumption, which implies that $\{\hat{\mathbf{w}}_k\}_{k=1}^K$ are not the optimal solutions. Therefore, the proof is complete.
\end{proof}

\subsection{Maximum-Ratio Transmission}
Given the PAs' positions $\mathbf{P}$, the MRT beamformer can be obtained according to \cite{heath2018foundations} as follows:
\begin{align}\label{DL_MRT}
{\mathbf{W}_{\rm{MRT}}}&=\sqrt{\frac{P}{{\rm{tr}}({\mathbf{H}}_{\rm{dl}}({\mathbf{P}}){\mathbf{H}}_{\rm{dl}}^{\mathsf{H}}({\mathbf{P}}))}}{\mathbf{H}}_{\rm{dl}}^{\mathsf{H}}({\mathbf{P}}),
\end{align}
which satisfies the equality power constraint in \textbf{Lemma \ref{Equality power constraint}}. 
By substituting equation \eqref{DL_MRT} into \eqref{SINR_dl}, We obtain the following simplified SINR expression, as shown in \eqref{SINR_dl_MRT} at the top of next page.
\begin{figure*}[!t]
\normalsize
\begin{align}\label{SINR_dl_MRT}
\gamma_{{\rm{dl}},k}^{\rm{MRT}}(\mathbf{P})&=\frac{\|{\mathbf{h}}_k^{\mathsf{H}}({\mathbf{P}}){\mathbf{G}}_{\rm{dl}}({\mathbf{P}})\|^4}
{\sum_{k'\ne k}\lvert{\mathbf{h}}_k^{\mathsf{H}}({\mathbf{P}}){\mathbf{G}}_{\rm{dl}}({\mathbf{P}}){\mathbf{G}}^{\mathsf{H}}_{\rm{dl}}({\mathbf{P}}){\mathbf{h}}_{k'}({\mathbf{P}})\rvert^2+\frac{{\rm{tr}}({\mathbf{H}}_{\rm{dl}}({\mathbf{P}}){\mathbf{H}}_{\rm{dl}}^{\mathsf{H}}({\mathbf{P}}))\sigma_k^2 }{P}}.
\end{align}
\hrulefill
\end{figure*}

Therefore, the joint optimization problem in \eqref{DL_Sum_Rate_Problem} can be simplified into the following problem with respect to $\mathbf{P}$:
\begin{subequations}\label{DL_Sum_Rate_Problem_MRT}
\begin{align}
\max_{{\mathbf{P}}}&~{\mathcal{R}}_{\rm{dl}}^{\rm{MRT}}({\mathbf{P}})=\sum_{k=1}^{K}\log_2(1+\gamma_{{\rm{dl}},k}^{\rm{MRT}}(\mathbf{P}))\label{DL_Sum_Rate_Problem_MRT_obj}\\
{\rm{s.t.}}&~\eqref{dl_c1},~\eqref{dl_c2},
\end{align}
\end{subequations}

Observing problem \eqref{DL_Sum_Rate_Problem_MRT}, it can be found that all optimization variables $\{p_{m,n}, \forall m,n\}$ are coupled, making it challenging to obtain the optimal solution directly. Therefore, an element-wise sequential optimization method with one-dimensional search can be employed to find a locally optimal solution for $\mathbf{P}$. Specifically, the position $p_{m,n}$ can be determined via a one-dimensional search, while keeping all other PAs' positions fixed.
Then, the objective function with respect to the single variable $p_{mn}$ can be reformulated as follows:
\begin{align}\label{DL_MRT_simplest}
{\mathcal{R}}_{\rm{dl}}^{\rm{MRT}}(p_{m,n})&=\sum_{k=1}^{K}\log_2\left(1+\frac{\mathcal{A}_k(p_{m,n})}{\mathcal{B}_k(p_{m,n})-\mathcal{A}_k(p_{m,n})}\right).
\end{align}
Here, $\mathcal{A}_k(p_{m,n})$ and $\mathcal{B}_k(p_{m,n})$ are given as follows:
\begin{subequations}
\begin{align}
&\mathcal{A}_k(p_{m,n})=\lvert a_{k,k}(\mathbf{p}_m)\rvert^2+2a_{k,k}(\mathbf{p}_m)c_{k,k}+\lvert c_{k,k}\rvert^2,\\
&\mathcal{B}_k(p_{m,n})=\sum_{k'=1}^{K}\Big[\lvert a_{k,k'}(\mathbf{p}_m)\rvert^2+2\Re\{a_{k,k'}(\mathbf{p}_m)c_{k,k'}^\mathsf{H}\}\nonumber\\
&\quad\quad\quad\quad+\left(a_{k',k'}(\mathbf{p}_m)+c_{k',k'}\right)\frac{\sigma_k^2}{P}+\lvert c_{k,k'}\rvert^2\Big].
\end{align}
\end{subequations}
and the term $a_{k,k'}(\mathbf{p}_m)=\bar{a}_k(p_{m,n})\bar{a}_{k'}^\mathsf{H}(p_{m,n})$, $c_{k,k'}=\sum_{m'\neq m}^{M}a_{k,k'}(\mathbf{p}_{m'})$, where $\bar{a}_k(p_{m,n})$ is given as follows:
\begin{align}\label{apmn}
\bar{a}_k(p_{m,n})=\Pi_k(p_{m,n})+\sum_{n'\neq n}^{N}\Pi_k(p_{m,n'}),
\end{align}
with 
\begin{align}\label{channel_all_PAs}
\Pi_k(p_{m,n}) = \frac{\lambda\sqrt{\eta_{m,n}}}{4\pi d_k(p_{m,n})}e^{-j[k_0 d_k(p_{m,n}) + k_{\rm{g}}(p_{m,n}-o_m)]}.
\end{align}

By discretizing the entire waveguide length $L_m$ into $N_s$ sampling points with a segment length of $\delta_m = \frac{L_m}{N_s-1}$, the set of candidate positions for each waveguide can be defined as follows:
\begin{equation}
\mathcal{X}_m \triangleq \left\{o_m + i\delta_m \;\middle|\; i \in \mathcal{I}_m \right\},\forall m\in[M],
\end{equation}
where $\mathcal{I}_m =\{ 0,1,\dots,N_s - 1\}$ denotes the all candidate indices. As the number of sampling points approaches infinity, PAs tend to be continuously activated.
An approximate optimal position $p_{m,n}$ can be obtained by selecting  
\begin{equation}\label{position_MRT}
p_{m,n}= \arg\max_{p_{m,n} \in \mathcal{X}_m/\mathcal{X}_m(\hat{\mathcal{I}}_m)} {\mathcal{R}}_{\rm{dl}}^{\rm{MRT}}(p_{m,n}),
\end{equation}
where $/$ denotes the set difference, and $\mathcal{X}_m(\hat{\mathcal{I}}_m)$ represents the set $\mathcal{X}_m$ corresponding to the indices $i \in \hat{\mathcal{I}}_m$. The set $\hat{\mathcal{I}}_m$ is defined as follows:
\begin{equation}
\hat{\mathcal{I}}_m=
\begin{cases}
  \displaystyle
  \mathcal{I}_m \cap\bigcup_{n'=1}^{n-1} 
  \left\{ i \;\middle|\;
  i \in \left\{ i_{m,n'}^{\rm floor}, 
  \dots, 
  i_{m,n'}^{\rm ceil} \right\} 
  \right\}, & n \neq 1 \\
  \varnothing, & n=1,
\end{cases}
\end{equation}
where $i_{m,n'}^{\rm floor}=\left\lfloor \frac{p_{m,n'}-o_m-\Delta}{\delta_m} \right\rfloor$, $i_{m,n'}^{\rm ceil}=\left\lceil \frac{p_{m,n'}-o_m+\Delta}{\delta_m} \right\rceil$, ensuring that the distance between each PA on the same waveguide satisfies the constraint \eqref{dl_c1}. By using \eqref{position_MRT} to sequentially update $\{p_{m,n}, \forall m,n\}$, a locally optimal solution to \eqref{DL_Sum_Rate_Problem_MRT} can be obtained. 

Taking the downlink MRT baseband beamforming as an example, the proposed element-wise sequential optimization procedure is summarized in \textbf{Algorithm~\ref{Algorithm}}. The subsequent algorithms follow a similar procedure, with simple modifications to lines 4 and 6 of \textbf{Algorithm~\ref{Algorithm}}.

\begin{algorithm}[!t]
\caption{\textbf{The Proposed Low-complexity Pinching Beamforming Algorithm}}\label{Algorithm}
\begin{algorithmic}[1]
\State Initialize the positions of PAs $\mathbf{P}$.
\Repeat
\For{$m \in \{1, \ldots, M\}$}
		\State Compute the $\mathbf{p}_m$-independent constant part of \Statex \hspace{\algorithmicindent}\hspace{\algorithmicindent}\eqref{DL_MRT_simplest}: $c_{k,k'}, \forall k,k'\in{\{1,\ldots,K\}}$.
			\For{$n \in \{1, \ldots, N\}$ }
        \State Update $p_{m,n}$ by solving problem \eqref{position_MRT} through \Statex \hspace{\algorithmicindent}\hspace{\algorithmicindent}\hspace{\algorithmicindent}one-dimensional search.
    		\EndFor
\EndFor
\Until{the fractional decrease of the objective value \eqref{DL_Sum_Rate_Problem_MRT_obj} \Statex \hspace{\algorithmicindent}falls below a predefined threshold}.
\State Calculate ${\mathbf{W}_{\rm{MRT}}}$ as \eqref{DL_MRT}.
\end{algorithmic}
\end{algorithm}

\subsection{Zero-Forcing}
In this case, given the PAs' positions $\mathbf{P}$, the power-normalized ZF beamformer can be obtained according to \cite{heath2018foundations} as follows:
\begin{align}\label{DL_ZF}
{\mathbf{W}_{\rm{ZF}}}&=\sqrt{\frac{P}{{\rm{tr}}(({\mathbf{H}}_{\rm{dl}}({\mathbf{P}}){\mathbf{H}}_{\rm{dl}}^{\mathsf{H}}({\mathbf{P}}))^{-1})}}{\mathbf{H}}_{\rm{dl}}^{\mathsf{H}}({\mathbf{P}})({\mathbf{H}}_{\rm{dl}}({\mathbf{P}}){\mathbf{H}}_{\rm{dl}}^{\mathsf{H}}({\mathbf{P}}))^{-1}.
\end{align}
Substituting equation \eqref{DL_ZF} into \eqref{SINR_dl} yields the simplified SINR expression as follows:
\begin{align}\label{SINR_dl_ZF}
\gamma_{{\rm{dl}},k}^{\rm{ZF}}(\mathbf{P})=\frac{P}{{\rm{tr}}(({\mathbf{H}}_{\rm{dl}}({\mathbf{P}}){\mathbf{H}}_{\rm{dl}}^{\mathsf{H}}({\mathbf{P}}))^{-1}){\sigma_k^2}}.
\end{align}
Therefore, the optimization problem in \eqref{DL_Sum_Rate_Problem} can be simplified as follows:
\begin{subequations}\label{DL_Sum_Rate_Problem_ZF}
\begin{align}
\max_{{\mathbf{P}}}&~{\mathcal{R}}_{\rm{dl}}^{\rm{ZF}}({\mathbf{P}})=\sum_{k=1}^{K}\log_2\left(1+\frac{P}{{\rm{tr}}(({\mathbf{H}}_{\rm{dl}}({\mathbf{P}}){\mathbf{H}}_{\rm{dl}}^{\mathsf{H}}({\mathbf{P}}))^{-1}){\sigma_k^2}}\right)\\
{\rm{s.t.}}&~\eqref{dl_c1},~\eqref{dl_c2},
\end{align}
\end{subequations}
Observing problem \eqref{DL_Sum_Rate_Problem_ZF}, since the summation terms are independent and the logarithm function is monotonically increasing, the problem is equivalently transformed into the following form:
\begin{subequations}\label{DL_Sum_Rate_Problem_ZF_1}
\begin{align}
\min_{{\mathbf{P}}}&~{\rm{tr}}(({\mathbf{H}}_{\rm{dl}}({\mathbf{P}}){\mathbf{H}}_{\rm{dl}}^{\mathsf{H}}({\mathbf{P}}))^{-1})\\
{\rm{s.t.}}&~\eqref{dl_c1},~\eqref{dl_c2}.
\end{align}
\end{subequations}
Similarly, we adopt the element-wise sequential optimization method to optimize the position $p_{m,n}$, while keeping the remaining elements in $\mathbf{P}$ fixed.
Let $\breve{\mathbf{h}}(\mathbf{p}_m)\in\mathbb{C}^{K\times 1}$ denote the $m$th column of ${\mathbf{H}}_{\rm{dl}}({\mathbf{P}})$ and the $k$th entry of $\breve{\mathbf{h}}(\mathbf{p}_m)$ is $[\breve{\mathbf{h}}(\mathbf{p}_m)]_k=\bar{a}_k(p_{m,n})$ as defined in \eqref{apmn}. Then, the objective function in \eqref{DL_Sum_Rate_Problem_ZF_1} can be decomposed as follows:
\begin{align}\label{ZF_decomposed}
{\rm{tr}}(({\mathbf{H}}_{\rm{dl}}({\mathbf{P}}){\mathbf{H}}_{\rm{dl}}^{\mathsf{H}}({\mathbf{P}}))^{-1})={\rm{tr}}((\breve{\mathbf{h}}(\mathbf{p}_m)\breve{\mathbf{h}}^{\mathsf{H}}(\mathbf{p}_m)+\breve{\mathbf{H}}_{\rm{dl}})^{-1}),
\end{align}
where $\breve{\mathbf{H}}_{\rm{dl}}=\sum_{m'\neq m}^{M}\breve{\mathbf{h}}(\mathbf{p}_{m'})\breve{\mathbf{h}}^{\mathsf{H}}(\mathbf{p}_{m'})$.
To avoid the high computational complexity caused by computing the matrix inverse at each sampling point, we can invoke the Sherman-Morrison lemma to further reformulate \eqref{ZF_decomposed} as follows:
\begin{align}
&{\rm{tr}}((\breve{\mathbf{h}}(\mathbf{p}_m)\breve{\mathbf{h}}^{\mathsf{H}}(\mathbf{p}_m)+\breve{\mathbf{H}}_{\rm{dl}})^{-1})\nonumber\\
&\quad\quad\quad={\rm{tr}}\left(\breve{\mathbf{H}}_{\rm{dl}}^{-1}-\frac{\breve{\mathbf{H}}_{\rm{dl}}^{-1}\breve{\mathbf{h}}(\mathbf{p}_m)\breve{\mathbf{h}}^{\mathsf{H}}(\mathbf{p}_m)\breve{\mathbf{H}}_{\rm{dl}}^{-1}}{1+\breve{\mathbf{h}}^{\mathsf{H}}(\mathbf{p}_m)\breve{\mathbf{H}}_{\rm{dl}}^{-1}\breve{\mathbf{h}}(\mathbf{p}_m)}\right).
\end{align}
\begin{remark}\label{remark1}
Due to the random nature of wireless channels, the independence of $\{\breve{\mathbf{h}}(\mathbf{p}_m)\}_{m=1}^{M}$ naturally holds. Moreover, we assume $M > K$ to ensure that each user has sufficient spatial degrees of freedom for space-division multiple access (SDMA), which implies that $\breve{\mathbf{H}}_{\rm{dl}}({\mathbf{P}})$ is full-rank. To this end, the Sherman-Morrison lemma can be employed.
\end{remark}

Therefore, a high-quality position $p_{m,n}$ can be obtained by solving the following problem using the same method as \eqref{position_MRT}:
\begin{align}\label{position_ZF}
p_{m,n}= \arg\max_{p_{m,n} \in \mathcal{X}_m/\mathcal{X}_m(\hat{\mathcal{I}}_m)}
\frac{\breve{\mathbf{h}}^{\mathsf{H}}(\mathbf{p}_m)\breve{\mathbf{H}}_{\rm{dl}}^{-2}\breve{\mathbf{h}}(\mathbf{p}_m)}{1+\breve{\mathbf{h}}^{\mathsf{H}}(\mathbf{p}_m)\breve{\mathbf{H}}_{\rm{dl}}^{-1}\breve{\mathbf{h}}(\mathbf{p}_m)}.
\end{align}
By iteratively updating $\{p_{m,n}, \forall m,n\}$ using \eqref{position_ZF}, a locally optimal solution to \eqref{DL_Sum_Rate_Problem_ZF} can be achieved.

\subsection{Minimum Mean Square Error}
Implementing ZF beamforming requires the number of users $K$ to be no larger than the number of RF chains $M$, i.e., $K\leq M$. In overloaded scenarios where $K>M$, it becomes infeasible to nullify all inter-user interference. On the other hand, when the locations of all PAs are fixed, finding an optimal solution for the linear beamforming matrix ${\mathbf{W}}$ is a challenging task, let alone deriving a closed-form expression. Fortunately, the authors in \cite{bjornson2014optimal} have proven that transmit MMSE beamforming \cite{joham2005linear} follows the optimal beamforming structure and generally achieves promising performance. More importantly, this MMSE beamformer can be designed regardless of the relationship between $K$ and $M$, which is given by
\begin{align}\label{DL_MMSE}
{\mathbf{W}}=\sqrt{\beta({\mathbf{P}})}{\mathbf{H}}_{\rm{dl}}^{\mathsf{H}}({\mathbf{P}})\left(\frac{P}{K\sigma^2}{\mathbf{H}}_{\rm{dl}}({\mathbf{P}}){\mathbf{H}}_{\rm{dl}}^{\mathsf{H}}({\mathbf{P}})+{\mathbf{I}}_K\right)^{-1},
\end{align}
where $\beta({\mathbf{P}})=\frac{P}{{\rm{tr}}\left((\frac{P}{K\sigma^2}{\mathbf{H}}_{\rm{dl}}({\mathbf{P}}){\mathbf{H}}_{\rm{dl}}^{\mathsf{H}}({\mathbf{P}})+{\mathbf{I}}_K)^{-2}{\mathbf{H}}_{\rm{dl}}({\mathbf{P}}){\mathbf{H}}_{\rm{dl}}^{\mathsf{H}}({\mathbf{P}})\right)}$ denotes the power normalization factor.
Substituting equation \eqref{DL_MMSE} into \eqref{SINR_dl} yields the following expression\eqref{SINR_dl_MMSE}, as shown at the top of the next page,
\begin{figure*}[!t]
\vspace{-10pt}
\begin{align}\label{SINR_dl_MMSE}
\gamma_{{\rm{dl}},k}^{\rm{MMSE}}(\mathbf{P})&=\frac{\lvert{\mathbf{h}}_k^{\mathsf{H}}({\mathbf{P}}){\mathbf{G}}_{\rm{dl}}({\mathbf{P}})\mathbf{M}({\mathbf{P}}){\mathbf{G}}_{\rm{dl}}^{\mathsf{H}}({\mathbf{P}}){\mathbf{h}}_k({\mathbf{P}})\rvert^2}
{\sum_{k'\ne k}\lvert{\mathbf{h}}_k^{\mathsf{H}}({\mathbf{P}}){\mathbf{G}}_{\rm{dl}}({\mathbf{P}})\mathbf{M}({\mathbf{P}}){\mathbf{G}}_{\rm{dl}}^{\mathsf{H}}({\mathbf{P}}){\mathbf{h}}_{k'}({\mathbf{P}})\rvert^2+\frac{\sigma^2}{\beta({\mathbf{P}})}}\nonumber\\
&=\frac{{\mathbf{H}}_{\mathrm{dl}{| k}}({\mathbf{P}})\mathbf{M}({\mathbf{P}}){\mathbf{H}}_{\mathrm{dl}{| k}}^{\mathsf{H}}({\mathbf{P}}){\mathbf{H}}_{\mathrm{dl}{| k}}({\mathbf{P}})\mathbf{M}({\mathbf{P}}){\mathbf{H}}_{\mathrm{dl}{| k}}^{\mathsf{H}}({\mathbf{P}})}
{{\mathbf{H}}_{\mathrm{dl}{| k}}({\mathbf{P}})\mathbf{M}({\mathbf{P}}){\mathbf{H}}_{\mathrm{dl}{\backslash k}}^{\mathsf{H}}({\mathbf{P}}){\mathbf{H}}_{\mathrm{dl}{\backslash k}}({\mathbf{P}})\mathbf{M}({\mathbf{P}}){\mathbf{H}}_{\mathrm{dl}{| k}}^{\mathsf{H}}({\mathbf{P}})+\frac{\sigma^2}{\beta({\mathbf{P}})}},
\end{align}
\hrulefill
\end{figure*}
where ${\mathbf{H}}_{\mathrm{dl}{| k}}({\mathbf{P}})\in \mathbb{C}^{1\times M}$ and ${\mathbf{H}}_{\mathrm{dl}{\backslash k}}({\mathbf{P}})\in \mathbb{C}^{(K-1)\times M}$ represent the extraction of only the $k$-th row of ${\mathbf{H}}_{\mathrm{dl}}({\mathbf{P}})$ and the removal of the $k$-th row of ${\mathbf{H}}_{\mathrm{dl}}({\mathbf{P}})$, respectively, and $\mathbf{M}({\mathbf{P}})=\left(\frac{P}{K\sigma^2}{\mathbf{H}}_{\rm{dl}}^{\mathsf{H}}({\mathbf{P}}){\mathbf{H}}_{\rm{dl}}({\mathbf{P}})+{\mathbf{I}}_M\right)^{-1}$.
Similarly, to avoid matrix inversion at each sampling point, we first apply the Woodbury matrix identity to transform $\mathbf{M}({\mathbf{P}})$ into:
\begin{align}
&\mathbf{M}({\mathbf{P}})={\mathbf{I}}_M \nonumber \\
&\quad-\frac{P}{K\sigma^2}{\mathbf{H}}_{\rm{dl}}^{\mathsf{H}}({\mathbf{P}})\left({\mathbf{I}}_K+\frac{P}{K\sigma^2}{\mathbf{H}}_{\rm{dl}}({\mathbf{P}}){\mathbf{H}}_{\rm{dl}}^{\mathsf{H}}({\mathbf{P}})\right)^{-1}{\mathbf{H}}_{\rm{dl}}({\mathbf{P}}).
\end{align}
Then, by applying the Sherman-Morrison formula, the matrix inversion operation can be eliminated when updating the position $p_{m,n}$ using the following equation:
\begin{align}\label{inverse_eliminated}
&\left({\mathbf{I}}_K+\frac{P}{K\sigma^2}{\mathbf{H}}_{\rm{dl}}({\mathbf{P}}){\mathbf{H}}_{\rm{dl}}^{\mathsf{H}}({\mathbf{P}})\right)^{-1}\nonumber\\
&\quad\quad\quad\quad\quad=\hat{\mathbf{H}}_{\rm{dl}}^{-1}-\frac{\hat{\mathbf{H}}_{\rm{dl}}^{-1}\breve{\mathbf{h}}(\mathbf{p}_m)\breve{\mathbf{h}}^{\mathsf{H}}(\mathbf{p}_m)\hat{\mathbf{H}}_{\rm{dl}}^{-1}}{\frac{K\sigma^2}{P}+\breve{\mathbf{h}}^{\mathsf{H}}(\mathbf{p}_m)\hat{\mathbf{H}}_{\rm{dl}}^{-1}\breve{\mathbf{h}}(\mathbf{p}_m)},
\end{align}
where $\breve{\mathbf{h}}(\mathbf{p}_m)\in\mathbb{C}^{K\times 1}$ also denotes the $m$th column of ${\mathbf{H}}_{\rm{dl}}({\mathbf{P}})$ and $\hat{\mathbf{H}}_{\rm{dl}}$ is defined as $\hat{\mathbf{H}}_{\rm{dl}}=\mathbf{I}_K + \frac{P}{K\sigma^2}\sum_{m'\neq m}^{M}\breve{\mathbf{h}}(\mathbf{p}_{m'})\breve{\mathbf{h}}^{\mathsf{H}}(\mathbf{p}_{m'})$. The same operation in \eqref{inverse_eliminated} can also be applied to $\beta({\mathbf{P}})$.

Therefore, the joint optimization problem in \eqref{DL_Sum_Rate_Problem} can be simplified into the following problem with respect to $\mathbf{P}$:
\begin{subequations}\label{DL_Sum_Rate_Problem_MRT}
\begin{align}
\max_{{\mathbf{P}}}&~{\mathcal{R}}_{\rm{dl}}^{\rm{MMSE}}({\mathbf{P}})=\sum_{k=1}^{K}\log_2(1+\gamma_{{\rm{dl}},k}^{\rm{MMSE}}(\mathbf{P}))\\
{\rm{s.t.}}&~\eqref{dl_c1},~\eqref{dl_c2}.
\end{align}
\end{subequations}
For the update of each PA's position $p_{m,n}$, it suffices to compute $\breve{\mathbf{h}}(\mathbf{p}_m)$ at each sampling point, and then a near-optimal position can be selected by solving the following problem
\begin{equation}\label{position_MRT}
p_{m,n}= \arg\max_{p_{m,n} \in \mathcal{X}_m/\mathcal{X}_m(\hat{\mathcal{I}}_m)} {\mathcal{R}}_{\rm{dl}}^{\rm{MMSE}}(\mathbf{P}).
\end{equation}

\section{Uplink Beamforming Design}
Following the methodology used in downlink PASS, we consider that the receive beamforming is designed using classical linear approaches, including MRC, ZF, and MMSE. For convenience, we define the effective channel as:
\begin{subequations}
\begin{align}
{\mathbf{H}}_{\rm{ul}}({\mathbf{P}})&=[{\mathbf{G}}_{\rm{ul}}^{\mathsf{T}}({\mathbf{P}}){\mathbf{h}}_1^{*}({\mathbf{P}}),\ldots,{\mathbf{G}}_{\rm{ul}}^{\mathsf{T}}({\mathbf{P}}){\mathbf{h}}_K^{*}({\mathbf{P}})]\\
&={\mathbf{G}}_{\rm{ul}}^{\mathsf{T}}({\mathbf{P}})[{\mathbf{h}}_1^{*}({\mathbf{P}}),\ldots,
{\mathbf{h}}_K^{*}({\mathbf{P}})]\in{\mathbbmss{C}}^{M\times K}.
\end{align}
\end{subequations}
\subsection{Maximum-Ratio Combining}
Given the PAs' positions $\mathbf{P}$, the MRC receive beamforming is given as follows\cite{heath2018foundations}:
\begin{align}\label{UL_MRC}
{\mathbf{V}}={\mathbf{H}}_{\rm{ul}}({\mathbf{P}}).
\end{align}
By substituting equation \eqref{UL_MRC} into \eqref{per_user_sinr}, we obtain the following simplified SINR expression \eqref{per_user_sinr_MRC}, as shown at the top of the next page.
\begin{figure*}[!t]
\vspace{-10pt}
\begin{align}\label{per_user_sinr_MRC}
\gamma_{{\rm{ul}},k}^{\rm{MRC}}({\mathbf{P}})=\frac{P_k\|{\mathbf{G}}_{\rm{ul}}^{\mathsf{T}}({\mathbf{P}}){\mathbf{h}}_k^{*}({\mathbf{P}})\|^4}
{{\mathbf{h}}_k^{\mathsf{T}}({\mathbf{P}}){\mathbf{G}}_{\rm{ul}}^{*}({\mathbf{P}})\left(\sum_{k'\ne k}P_{k'}{\mathbf{G}}_{\rm{ul}}^{\mathsf{T}}({\mathbf{P}}){\mathbf{h}}_{k'}^{*}({\mathbf{P}}){\mathbf{h}}_{k'}^{\mathsf{T}}({\mathbf{P}}){\mathbf{G}}_{\rm{ul}}^{*}({\mathbf{P}})+{\mathbf{R}}_{\mathbf{z}}({\mathbf{P}})\right){\mathbf{G}}_{\rm{ul}}^{\mathsf{T}}({\mathbf{P}}){\mathbf{h}}_k^{*}({\mathbf{P}})}.
\end{align}
\hrulefill
\end{figure*}
Therefore, the joint optimization problem in \eqref{UL_Sum_Rate_Problem} can be simplified into the following problem with respect to $\mathbf{P}$:
\begin{subequations}\label{UL_Sum_Rate_Problem_MRC}
\begin{align}
\max_{{\mathbf{P}}}&~{\mathcal{R}}_{\rm{ul}}^{\rm{MRC}}({\mathbf{P}})=\sum_{k=1}^{K}\log_2(1+\gamma_{{\rm{ul}},k}^{\rm{MRC}}(\mathbf{P}))\\
{\rm{s.t.}}&~\eqref{ul_c1},~\eqref{ul_c2}.
\end{align}
\end{subequations}
When updating $p_{m,n}$ while keeping the other PAs' positions fixed, $\gamma_{{\rm{ul}},k}^{\rm{MRC}}(\mathbf{P})$ can be reformulated as follows:
\begin{align}\label{per_user_sinr_MRC_1}
\gamma_{{\rm{ul}},k}^{\rm{MRC}}(p_{m,n})=\frac{P_k\mathcal{A}_k(p_{m,n})}
{\mathcal{D}_k(p_{m,n})+C}.
\end{align}
where $\mathcal{D}_k(p_{m,n})$ is defined as \eqref{Dk}, as shown at the top of the next page, 
\begin{figure*}[!t]
\vspace{-15pt}
\begin{align}\label{Dk}
\mathcal{D}_k(p_{m,n}) = \sum_{m'=1}^M\sum_{k'\neq k}P_{k'} {a}_{k,k'}^{\mathsf{H}}(\mathbf{p}_{m}){a}_{k,k'}(\mathbf{p}_{m'})+\sum_{m'\neq m}\sum_{k'\neq k}P_{k'}{a}_{k,k'}^{\mathsf{H}}(\mathbf{p}_{m'})a_{k,k'}(\mathbf{p}_{m})+ a_{k,k}(\mathbf{p}_{m})\sigma^2\|\mathbf{g}_{\rm{ul}}(\mathbf{p}_{m})\|^2.
\end{align}
\hrulefill
\end{figure*}
$C = \sum_{m'\neq m}\sum_{i\neq m}\sum_{k'\neq k}P_{k'} a_{k,k'}^{\mathsf{H}}(\mathbf{p}_{i})a_{k,k'}(\mathbf{p}_{m'})+\sum_{m'\neq m}a_{k,k}(\mathbf{p}_{m'})\sigma^2\|\mathbf{g}_{\rm{ul}}(\mathbf{p}_{m'})\|^2$ is a constant that is independent of $p_{m,n}$. And the definitions of $\mathcal{A}_k(p_{m,n})$ and $a_{k,k'}(\mathbf{p}_{m})$ remain consistent with those in the downlink MRT section. The only difference is that in equation \eqref{channel_all_PAs}, $\eta_{m,n}$ is replaced by $\zeta_{m,n}$.

Therefore, an approximate optimal position $p_{m,n}$ can be obtained by selecting  
\begin{equation}\label{position_MRC}
p_{m,n}= \arg\max_{p_{m,n} \in \mathcal{X}_m/\mathcal{X}_m(\hat{\mathcal{I}}_m)} \sum_{k=1}^{K}\log_2(1+\gamma_{{\rm{ul}},k}^{\rm{MRC}}(p_{m,n})),
\end{equation}
By sequentially updating all PAs' positions, \eqref{UL_Sum_Rate_Problem_MRC} can be solved.
%
%

\subsection{Zero-Forcing}
In this case, given the PAs' positions $\mathbf{P}$, the ZF receive beamforming is given as follows\cite{heath2018foundations}:
\begin{align}\label{UL_ZF}
{\mathbf{V}}={\mathbf{H}}_{\rm{ul}}({\mathbf{P}})({\mathbf{H}}_{\rm{ul}}^{\mathsf{H}}({\mathbf{P}}){\mathbf{H}}_{\rm{ul}}({\mathbf{P}}))^{-1}.
\end{align}
By substituting equation \eqref{UL_ZF} into \eqref{per_user_sinr}, we obtain the following simplified SINR expression:
\begin{align}\label{per_user_sinr_ZF}
\gamma_{{\rm{ul}},k}^{\rm{ZF}}({\mathbf{P}})=\frac{P_k}{[\mathbf{M}({\mathbf{P}})\mathbf{M}^{\mathsf{H}}({\mathbf{P}})]_{k,k}}.
\end{align}
where $\mathbf{M}({\mathbf{P}})=({\mathbf{H}}_{\rm{ul}}^{\mathsf{H}}({\mathbf{P}}){\mathbf{H}}_{\rm{ul}}({\mathbf{P}}))^{-1}{\mathbf{H}}_{\rm{ul}}^{\mathsf{H}}({\mathbf{P}}){\mathbf{R}}_{\mathbf{z}}^{\frac{1}{2}}({\mathbf{P}})$ and $[\cdot]_{k,k}$ denotes the $k$th row and the $k$th column of a matrix.
Therefore, the optimization problem in \eqref{UL_Sum_Rate_Problem} can be simplified as follows:
\begin{subequations}\label{UL_Sum_Rate_Problem_ZF_1}
\begin{align}
\max_{{\mathbf{P}}}&~{\mathcal{R}}_{\rm{ul}}^{\rm{ZF}}({\mathbf{P}})=\sum_{k=1}^{K}\log_2\left(1+\frac{P_k}{[\mathbf{M}({\mathbf{P}})\mathbf{M}^{\mathsf{H}}({\mathbf{P}})]_{k,k}}\right)\\
{\rm{s.t.}}&~\eqref{ul_c1},~\eqref{ul_c2},
\end{align}
\end{subequations}
Similarly, we adopt the element-wise method to sequentially optimize the position $p_{m,n}$, while keeping the remaining elements in $\mathbf{P}$ fixed.
Let $\breve{\mathbf{h}}(\mathbf{p}_m)\in\mathbb{C}^{K\times 1}$ denote the $m$th column of ${\mathbf{H}}_{\rm{ul}}^{\mathsf{H}}({\mathbf{P}})$ and the $k$th entry of $\breve{\mathbf{h}}(\mathbf{p}_m)$ is $[\breve{\mathbf{h}}(\mathbf{p}_m)]_k=\bar{a}_k^{\mathsf{H}}(p_{m,n})$ as defined in \eqref{apmn}. 
To avoid the matrix inverse at each sampling point, the Sherman-Morrison lemma is employed to further reformulate $\mathbf{M}({\mathbf{P}})$ as follows:
\begin{align}
\mathbf{M}({\mathbf{P}})=\left(\breve{\mathbf{H}}_{\rm{ul}}^{-1}-\frac{\breve{\mathbf{H}}_{\rm{ul}}^{-1}\breve{\mathbf{h}}(\mathbf{p}_m)\breve{\mathbf{h}}^{\mathsf{H}}(\mathbf{p}_m)\breve{\mathbf{H}}_{\rm{ul}}^{-1}}{1+\breve{\mathbf{h}}^{\mathsf{H}}(\mathbf{p}_m)\breve{\mathbf{H}}_{\rm{ul}}^{-1}\breve{\mathbf{h}}(\mathbf{p}_m)}\right)\mathbf{N}(\mathbf{p}_m),
\end{align}
where $\breve{\mathbf{H}}_{\rm{ul}}=\sum_{m'\neq m}^{M}\breve{\mathbf{h}}(\mathbf{p}_{m'})\breve{\mathbf{h}}^{\mathsf{H}}(\mathbf{p}_{m'})$ is invertible according to \textbf{Remark~\ref{remark1}} and $\mathbf{N}(\mathbf{p}_m)$ is given as follows:
\begin{align}
&\mathbf{N}(\mathbf{p}_m)=\sigma\Big[\lVert{\mathbf{g}}_{\rm{ul}}({\mathbf{p}}_{1})\rVert\breve{\mathbf{h}}(\mathbf{p}_1),\ldots,\lVert{\mathbf{g}}_{\rm{ul}}({\mathbf{p}}_{m})\rVert\breve{\mathbf{h}}(\mathbf{p}_m),\nonumber\\
&\quad\quad\quad\quad\quad\quad\quad\quad \ldots,\lVert{\mathbf{g}}_{\rm{ul}}({\mathbf{p}}_{M})\rVert\breve{\mathbf{h}}(\mathbf{p}_M)\Big]\in\mathbb{C}^{K\times M}.
\end{align}
Let $\mathbf{m}_k(\mathbf{p}_m)\in\mathbb{C}^{M\times 1} $ denote the $k$th column of the $\mathbf{M}^{\mathsf{H}}({\mathbf{P}})$, the per-user SINR can be further written as follows:
\begin{align}\label{per_user_sinr_ZF_1}
\gamma_{{\rm{ul}},k}^{\rm{ZF}}({\mathbf{p}_m})=\frac{P_k}{\|\mathbf{m}_k(\mathbf{p}_m)\|^2}.
\end{align}
Therefore, an approximate optimal position $p_{m,n}$ can be obtained by selecting  
\begin{equation}\label{position_MRC}
p_{m,n}= \arg\max_{p_{m,n} \in \mathcal{X}_m/\mathcal{X}_m(\hat{\mathcal{I}}_m)} \sum_{k=1}^{K}\log_2\left(1+\gamma_{{\rm{ul}},k}^{\rm{ZF}}({\mathbf{p}_m})\right).
\end{equation}
By sequentially updating all PAs' positions, \eqref{UL_Sum_Rate_Problem_ZF_1} can be solved.

\subsection{Minimum Mean-Squared Error}
The per-user SINR given in \eqref{per_user_sinr} features a generalized Rayleigh quotient, which is maximized when $\mathbf{V}\in{\mathbbmss{C}}^{M\times K}$ is set as follows \cite{heath2018foundations}:
\begin{align}\label{MMSE_Detector_First}
{\mathbf{V}} = ({\mathbf{H}}_{\rm{ul}}({\mathbf{P}}){\rm{diag}}(P_1;\ldots;P_K){\mathbf{H}}_{\rm{ul}}^{\mathsf{H}}({\mathbf{P}})+
{\mathbf{R}}_{\mathbf{z}}({\mathbf{P}}))^{-1}{\mathbf{H}}_{\rm{ul}}({\mathbf{P}}),
\end{align}
which is referred to as the optimal MMSE detector. Using this receive beamformer, the achievable rate for user $k$ can be written as \eqref{UL_per-user_Rate_MMSE}, as shown at the top of next page. 
\begin{figure*}[!t]
\begin{align}\label{UL_per-user_Rate_MMSE}
{\mathcal{R}}_{{\rm{ul}},k}^{\rm{MMSE}}({\mathbf{P}})&=\log_2\det({\mathbf{I}}_M+P_k({\mathbf{G}}_{\rm{ul}}^{\mathsf{T}}({\mathbf{P}}){\mathbf{h}}_k^{*}({\mathbf{P}}))
({\mathbf{G}}_{\rm{ul}}^{\mathsf{T}}({\mathbf{P}}){\mathbf{h}}_k^{*}({\mathbf{P}}))^{\mathsf{H}}(\sum_{k'\ne k}P_{k'}({\mathbf{G}}_{\rm{ul}}^{\mathsf{T}}({\mathbf{P}}){\mathbf{h}}_{k'}^{*}({\mathbf{P}}))
({\mathbf{G}}_{\rm{ul}}^{\mathsf{T}}({\mathbf{P}}){\mathbf{h}}_{k'}^{*}({\mathbf{P}}))^{\mathsf{H}}+
{\mathbf{R}}_{\mathbf{z}}({\mathbf{P}}))^{-1})\nonumber\\
&=\log_2\det(\sum_{k=1 }^{K}P_{k}({\mathbf{G}}_{\rm{ul}}^{\mathsf{T}}({\mathbf{P}}){\mathbf{h}}_{k}^{*}({\mathbf{P}}))
({\mathbf{G}}_{\rm{ul}}^{\mathsf{T}}({\mathbf{P}}){\mathbf{h}}_{k}^{*}({\mathbf{P}}))^{\mathsf{H}}+
{\mathbf{R}}_{\mathbf{z}}({\mathbf{P}}))\nonumber\\
&\quad\quad\quad\quad\quad\quad\quad\quad\quad\quad\quad\quad\quad\quad\quad\quad-\log_2\det(\sum_{k'\ne k}P_{k'}({\mathbf{G}}_{\rm{ul}}^{\mathsf{T}}({\mathbf{P}}){\mathbf{h}}_{k'}^{*}({\mathbf{P}}))
({\mathbf{G}}_{\rm{ul}}^{\mathsf{T}}({\mathbf{P}}){\mathbf{h}}_{k'}^{*}({\mathbf{P}}))^{\mathsf{H}}+
{\mathbf{R}}_{\mathbf{z}}({\mathbf{P}}))\nonumber\\
&=\log_2\det(\sum_{k=1 }^{K}P_{k}{\mathbf{R}}_{\mathbf{z}}^{-\frac{1}{2}}({\mathbf{P}})({\mathbf{G}}_{\rm{ul}}^{\mathsf{T}}({\mathbf{P}}){\mathbf{h}}_{k}^{*}({\mathbf{P}}))
({\mathbf{G}}_{\rm{ul}}^{\mathsf{T}}({\mathbf{P}}){\mathbf{h}}_{k}^{*}({\mathbf{P}}))^{\mathsf{H}}{\mathbf{R}}_{\mathbf{z}}^{-\frac{1}{2}}({\mathbf{P}})+{\mathbf{I}}_M)\nonumber\\
&\quad\quad\quad\quad\quad\quad\quad\quad\quad\quad-\log_2\det(\sum_{k'\ne k}P_{k'}{\mathbf{R}}_{\mathbf{z}}^{-\frac{1}{2}}({\mathbf{P}})({\mathbf{G}}_{\rm{ul}}^{\mathsf{T}}({\mathbf{P}}){\mathbf{h}}_{k'}^{*}({\mathbf{P}}))
({\mathbf{G}}_{\rm{ul}}^{\mathsf{T}}({\mathbf{P}}){\mathbf{h}}_{k'}^{*}({\mathbf{P}}))^{\mathsf{H}}{\mathbf{R}}_{\mathbf{z}}^{-\frac{1}{2}}({\mathbf{P}})+
{\mathbf{I}}_M).
\end{align}
\hrulefill
\end{figure*}
By applying Sylvester’s determinant identity, \eqref{UL_per-user_Rate_MMSE} can be further reformulated into \eqref{UL_per-user_Rate_MMSE_v1}, as shown at the top of the next page,
\begin{figure*}[!t]
\vspace{-10pt}
\begin{align}\label{UL_per-user_Rate_MMSE_v1}
{\mathcal{R}}_{{\rm{ul}},k}^{\rm{MMSE}}({\mathbf{P}})&=\log_2\det({\mathbf{H}}_{\rm{ul}}^{\mathsf{H}}({\mathbf{P}}){\mathbf{R}}_{\mathbf{z}}^{-1}({\mathbf{P}}){\mathbf{H}}_{\rm{ul}}({\mathbf{P}}){\rm{diag}}(P_1;\ldots;P_K)+{\mathbf{I}}_K)\nonumber\\
&-\log_2\det({\mathbf{H}}_{\mathrm{ul}{\backslash k}}^{\mathsf{H}}({\mathbf{P}}){\mathbf{R}}_{\mathbf{z}}^{-1}({\mathbf{P}}){\mathbf{H}}_{\mathrm{ul}{\backslash k}}({\mathbf{P}}){\rm{diag}}(P_1;\ldots;P_{k-1};P_{k+1};\ldots;P_K)+{\mathbf{I}}_{K-1}),
\end{align}
\hrulefill
\end{figure*}
where ${\mathbf{H}}_{\mathrm{ul}{\backslash k}}({\mathbf{P}})\in \mathbb{C}^{M\times (K-1)}$ is a reduced form of ${\mathbf{H}}_{\mathrm{ul}}({\mathbf{P}})$ with its $k$th column missing.

In this case, the optimization problem in \eqref{UL_Sum_Rate_Problem} can be simplified as follows:
\begin{subequations}\label{UL_Sum_Rate_Problem_MMSE}
\begin{align}
\max_{{\mathbf{P}}}&~{\mathcal{R}}_{\rm{ul}}^{\rm{MMSE}}({\mathbf{P}})=\sum_{k=1}^{K}{\mathcal{R}}_{{\rm{ul}},k}^{\rm{MMSE}}({\mathbf{P}})\\
{\rm{s.t.}}&~\eqref{ul_c1},~\eqref{ul_c2}.
\end{align}
\end{subequations}
Similar to the above, we can tackle the problem in \eqref{UL_Sum_Rate_Problem_MMSE} via the element-wise approach. To avoid determinant operations for each sample point, when updating $p_{m,n}$ while keeping the remaining PAs' positions fixed, we can add a constant term independent of $p_{m,n}$ to simplify the expression without affecting the optimization result. Consequently, by subtracting the constant term $\log_2\det(\mathbf{B})$, the first term in \eqref{UL_per-user_Rate_MMSE_v1} can be transformed as \eqref{term1}, where $\mathbf{B}$ is defined as $\mathbf{B}=\sum_{m'\neq m}\frac{1}{\sigma^2\lVert{\mathbf{g}}_{\rm{ul}}({\mathbf{p}}_{m'})\rVert^2}\breve{\mathbf{h}}(\mathbf{p}_{m'})\breve{\mathbf{h}}^{\mathsf{H}}(\mathbf{p}_{m'}){\rm{diag}}(P_1;\ldots;P_K)+{\mathbf{I}}_K$. 
Here, $\breve{\mathbf{h}}(\mathbf{p}_m)\in\mathbb{C}^{K\times 1}$ denotes the $m$th column of ${\mathbf{H}}_{\rm{ul}}^{\mathsf{H}}({\mathbf{P}})$ and the $k$th entry of $\breve{\mathbf{h}}(\mathbf{p}_m)$ is given by $[\breve{\mathbf{h}}(\mathbf{p}_m)]_k=\bar{a}_k^{\mathsf{H}}(p_{m,n})$ as defined in \eqref{apmn}. Similarly, let $\breve{\mathbf{h}}_{\backslash k}(\mathbf{p}_m)\in\mathbb{C}^{(K-1)\times 1}$ represent the $m$th column of ${\mathbf{H}}_{\mathrm{ul}{\backslash k}}^{\mathsf{H}}({\mathbf{P}})$, which is essentially $\breve{\mathbf{h}}(\mathbf{p}_m)$ with its $k$th entry removed.

\begin{figure*}[!t]
\vspace{-10pt}
\begin{align}\label{term1}
\mathrm{The~First~Term}:~&\log_2\det({\mathbf{H}}_{\rm{ul}}^{\mathsf{H}}({\mathbf{P}}){\mathbf{R}}_{\mathbf{z}}^{-1}({\mathbf{P}}){\mathbf{H}}_{\rm{ul}}({\mathbf{P}}){\rm{diag}}(P_1;\ldots;P_K)+{\mathbf{I}}_K)-\log_2\det(\mathbf{B})\nonumber\\
&\overset{(a)}{=}\log_2\det(\frac{1}{\sigma^2\lVert{\mathbf{g}}_{\rm{ul}}({\mathbf{p}}_{m})\rVert^2}\breve{\mathbf{h}}(\mathbf{p}_{m})\breve{\mathbf{h}}^{\mathsf{H}}(\mathbf{p}_{m}){\rm{diag}}(P_1;\ldots;P_K)\mathbf{B}^{-1}+{\mathbf{I}}_K)\nonumber\\
&\overset{(b)}{=}\log_2(1+\frac{1}{\sigma^2\lVert{\mathbf{g}}_{\rm{ul}}({\mathbf{p}}_{m})\rVert^2}\breve{\mathbf{h}}^{\mathsf{H}}(\mathbf{p}_{m}){\rm{diag}}(P_1;\ldots;P_K)\mathbf{B}^{-1}\breve{\mathbf{h}}(\mathbf{p}_{m})).
\end{align}
\hrulefill
\end{figure*}
With these definitions at hand, equality $(a)$ is obtained through some basic determinant operations, while equality $(b)$ follows from a further application of Sylvester’s determinant identity. Following similar operations, the second term in \eqref{UL_per-user_Rate_MMSE_v1} can be modified as \eqref{term2}, as shown at the top of next page,
\begin{figure*}[!t]
\begin{align}\label{term2}
\mathrm{The~Second~Term}:~&\log_2(1+\frac{1}{\sigma^2\lVert{\mathbf{g}}_{\rm{ul}}({\mathbf{p}}_{m})\rVert^2}\breve{\mathbf{h}}_{\backslash k}^{\mathsf{H}}(\mathbf{p}_{m}){\rm{diag}}(P_1;\ldots;P_{k-1};P_{k+1};\ldots;P_K)\bar{\mathbf{B}}^{-1}\breve{\mathbf{h}}_{\backslash k}(\mathbf{p}_{m})),
\end{align}
\hrulefill
\end{figure*}
where $\bar{\mathbf{B}}$ is defined as follows:
\begin{align}\label{BB}
\bar{\mathbf{B}}&=\sum_{m'\neq m}\frac{1}{\sigma^2\lVert{\mathbf{g}}_{\rm{ul}}({\mathbf{p}}_{m'})\rVert^2}\breve{\mathbf{h}}_{\backslash k}(\mathbf{p}_{m'})\breve{\mathbf{h}}_{\backslash k}^{\mathsf{H}}(\mathbf{p}_{m'})\nonumber\\
&\quad\quad\quad{\rm{diag}}(P_1;\ldots;P_{k-1};P_{k+1};\ldots;P_K)+{\mathbf{I}}_{K-1}.
\end{align}
Therefore, an approximate optimal position $p_{m,n}$ can be obtained by selecting  
\begin{equation}\label{position_MRC}
p_{m,n}= \arg\max_{p_{m,n} \in \mathcal{X}_m/\mathcal{X}_m(\hat{\mathcal{I}}_m)} \sum_{k=1}^{K}f_{\rm{ul}}^{\rm{MMSE}}(\mathbf{p}_{m}).
\end{equation}
where $f_{\rm{ul}}^{\rm{MMSE}}(\mathbf{p}_{m})$ is defined as \eqref{f_mmse}.
\begin{figure*}[!t]
\begin{align}\label{f_mmse}
f_{\rm{ul}}^{\rm{MMSE}}(\mathbf{p}_{m})&=\log_2(1+\frac{1}{\sigma^2\lVert{\mathbf{g}}_{\rm{ul}}({\mathbf{p}}_{m})\rVert^2}\breve{\mathbf{h}}^{\mathsf{H}}(\mathbf{p}_{m}){\rm{diag}}(P_1;\ldots;P_K)\mathbf{B}^{-1}\breve{\mathbf{h}}(\mathbf{p}_{m}))\nonumber\\
&-\log_2(1+\frac{1}{\sigma^2\lVert{\mathbf{g}}_{\rm{ul}}({\mathbf{p}}_{m})\rVert^2}\breve{\mathbf{h}}_{\backslash k}^{\mathsf{H}}(\mathbf{p}_{m}){\rm{diag}}(P_1;\ldots;P_{k-1};P_{k+1};\ldots;P_K)\bar{\mathbf{B}}^{-1}\breve{\mathbf{h}}_{\backslash k}(\mathbf{p}_{m})).
\end{align}
\hrulefill
\end{figure*}
By sequentially updating all PAs' positions, \eqref{UL_Sum_Rate_Problem_MMSE} can be solved.


\section{Numerical Investigations}
In this section, we conduct extensive simulations to evaluate the performance gain of the PASS over conventional fixed-antenna systems. As discussed in the previous sections, the number of sampling points is identified as a critical factor affecting the performance of PASS. Therefore, its impact on system performance is first examined. Subsequently, for both uplink and downlink, we analyze the effects of transmit power, the side length of the square region along the $x$-axis, and the number of PAs and users on system performance. Based on the results, several insightful observations are obtained.
\begin{figure*}[!t]
  \centering
  \setlength{\abovecaptionskip}{0pt}
  \subfigure[Achievable downlink sum-rate vs. search resolution $N_s$.]{
    \includegraphics[height=0.345\textwidth]{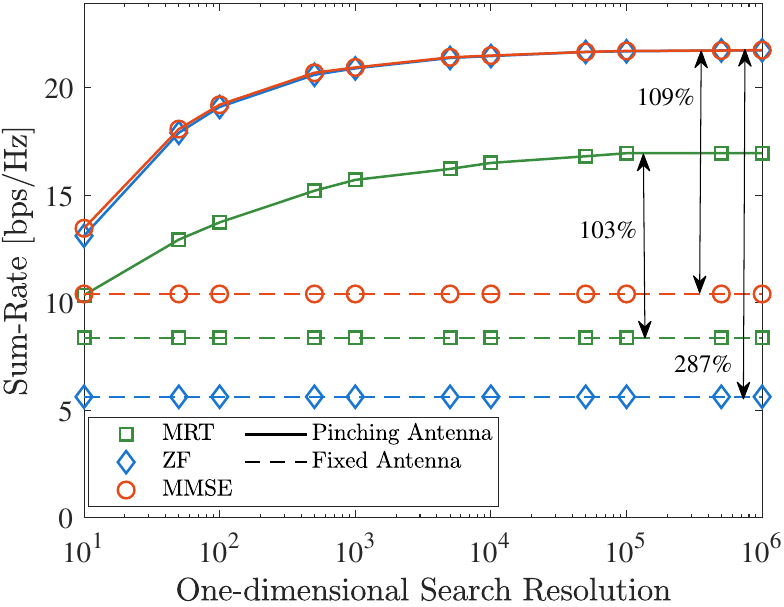}
    \label{DL_Ns}
  }
  \hspace{1cm}
  \subfigure[Achievable uplink sum-rate vs. search resolution $N_s$.]{
    \includegraphics[height=0.345\textwidth]{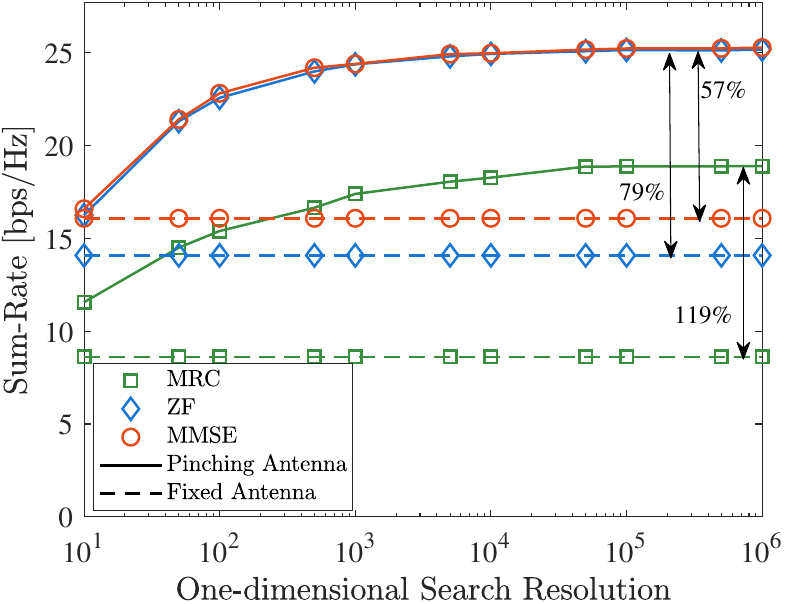}
    \label{UL_Ns}
  }
  \caption{Achievable sum-rate performance under different one-dimensional search resolutions.}
  \label{fig:DL_UL_Ns}
\end{figure*}
\begin{figure*}[!t]
  \centering
  \setlength{\abovecaptionskip}{0pt}
  \subfigure[Achievable downlink sum-rate vs. transmit power $P_d$.]{
    \includegraphics[height=0.35\textwidth]{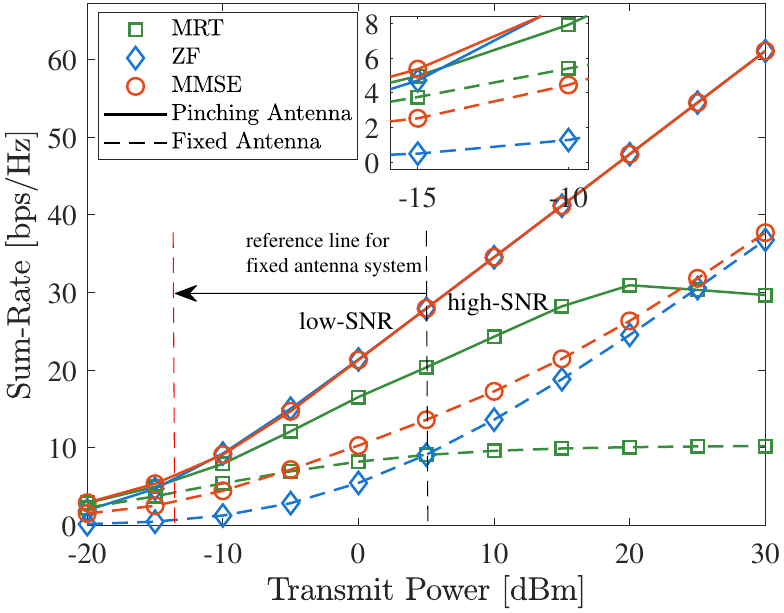}
    \label{DL_power}
  }
  \hspace{1cm}
  \subfigure[Achievable uplink sum-rate vs. per-user transmit power $P_u$.]{
    \includegraphics[height=0.35\textwidth]{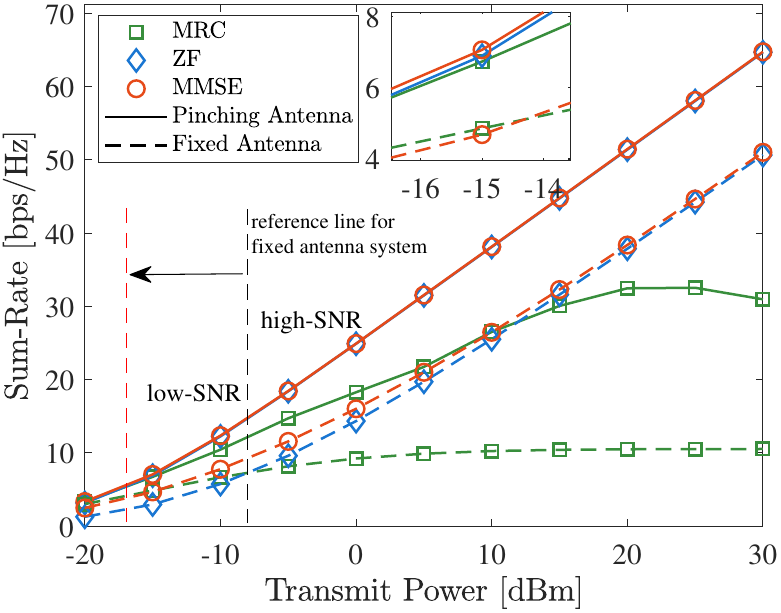}
    \label{UL_power}
  }
  \caption{Achievable sum-rate performance under different power levels.}
  \label{fig:DL_UL_power}
  \vspace{-10pt}
\end{figure*}
\subsection{Experimental Setting}
In the simulations, we consider a rectangular region located in the $x$-$y$ plane with its center at $[0, 0, 0]$, as illustrated in Fig. \ref{Figure: System_Model}. The length of each waveguide is equal to the side length of the service region along the $x$-axis, i.e., $L_m = D_x, \forall m \in [M]$, and the spacing between adjacent waveguides along the $y$-axis is given by $d = \frac{D_y}{M-1}$. Unless otherwise specified, the following parameters are adopted: the side lengths of the rectangular region are $ D_x = 50~\mathrm{m}$ and $D_y = 6~\mathrm{m}$. The BS is equipped with $M = 5$ waveguides deployed at a height $a = 5~\mathrm{m}$, each activating $N = 6$ PAs, and serves $K = 4$ users. The carrier frequency and noise power are set to $f = 28~\mathrm{GHz}$, $\sigma^2= -90~\mathrm{dBm}$. The waveguide wavelength is $\lambda_g = \frac{\lambda}{n_{\rm eff}}$, where $n_{\rm eff} = 1.4$ represents the effective refractive index of the dielectric waveguide. The average attenuation factor is $\kappa=0.1$. To prevent mutual coupling, the minimum spacing between PAs on the same waveguide is set to $\Delta = \frac{\lambda}{2}$. The BS transmit power and the per-user transmit power are $P_d = 0$ dBm and $P_u = 0$ dBm, respectively.For the one-dimensional search, the number of sampling points is set
to $N_s = 10^4$ and the positions of the PAs are initialized randomly. Numerical results are obtained by averaging 400 random seeds.

To ensure a fair comparison with the conventional fixed-antenna systems, a \textit{hybrid beamforming-based MIMO (hMIMO)} architecture is adopted as the benchmark. Specifically, the \textit{hMIMO} system employs a hybrid transceiver equipped with $M$ RF chains, where each RF chain is connected to $N$ antennas through a tunable phase shifters. Both downlink and uplink signal processing are performed using hybrid analog-digital schemes. The analog beamforming is implemented via phase shifter networks which follows the hybrid beamforming algorithms proposed in \cite{Zhang2024hybrid}, while the digital beamforming is designed in the baseband using the aforementioned methods. 
\begin{figure*}[!t]
  \centering
  \setlength{\abovecaptionskip}{0pt}
  \subfigure[Achievable downlink sum-rate vs. side length $D_x$.]{
    \includegraphics[height=0.35\textwidth]{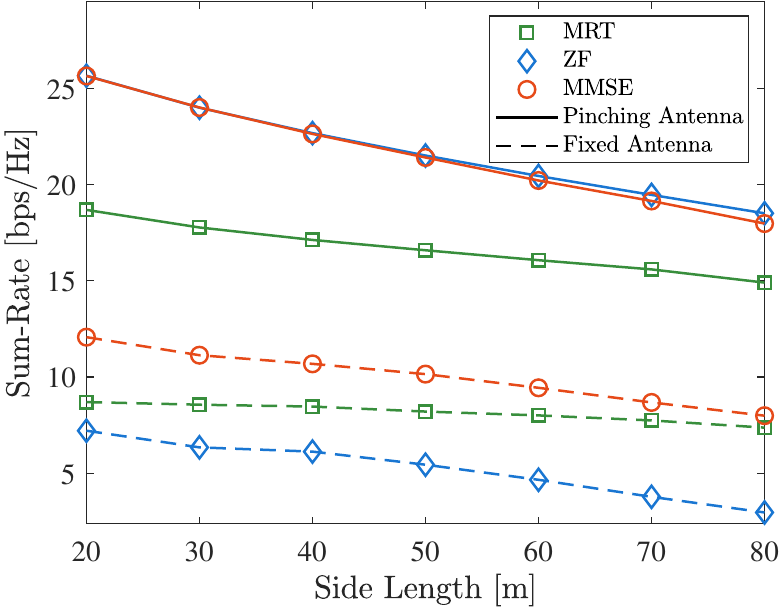}
    \label{DL_length}
  }
  \hspace{1cm}
  \subfigure[Achievable uplink sum-rate vs. side length $D_x$.]{
    \includegraphics[height=0.35\textwidth]{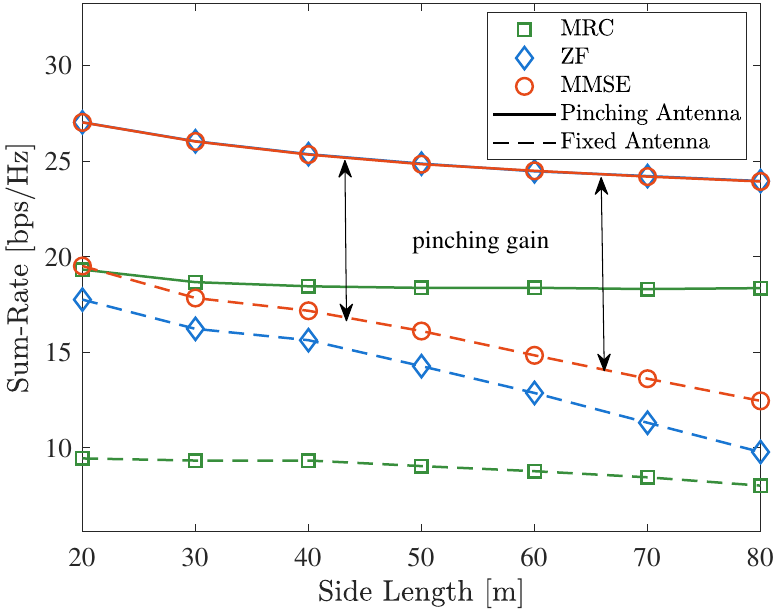}
    \label{UL_length}
  }
  \caption{Achievable sum-rate performance under different side length levels.}
  \label{fig:DL_UL_length}
\end{figure*}
\subsection{Sum-Rate Versus Search Resolution}
Fig. \ref{fig:DL_UL_Ns} presents the achievable sum-rate of the PASS as a function of the one-dimensional search resolution, for both downlink and uplink scenarios under various beamforming schemes. The one-dimensional search resolution refers to the number of candidate activation points along the waveguide, which directly affects the granularity of PAs' positions optimization in the element-wise process.
As observed in both figures, the sum-rate increases monotonically with search resolution across all considered beamforming schemes. This improvement stems from the fact that higher resolution enables more precise selection of PAs' positions, which allows better phase alignment with the user channel's spatial response and enhances the received signal power. Notably, when the search resolution reaches approximately $10^5$, the performance begins to saturate, indicating that near-optimal throughput can be achieved without requiring excessive computational complexity.
Furthermore, PASS exhibits significant performance advantages over conventional fixed-antenna systems in both downlink and uplink, even under reduced search resolutions. It thus serves as a promising solution for complexity-constrained system design.

\subsection{Sum-Rate Versus Transmit Power}
Fig. \ref{DL_power} illustrates the achievable downlink sum-rate as a function of the transmit power. As observed, under all beamforming schemes, PASS consistently outperforms the conventional fixed-antenna systems. This is because the flexible adjustment of PAs' positions in PASS facilitates signal phase control and reduces large-scale path loss, both of which contribute to improved system sum-rate performance.
In the fixed-antenna system, the three beamforming schemes exhibit well-known trends: MRT outperforms ZF in the low-SNR regime, while ZF surpasses MRT at high SNR, and MMSE performs well across the entire considered power range. This is attributed to the fact that MRT focuses solely on noise suppression, ZF targets only interference cancellation, whereas MMSE integrates both effects for joint suppression of noise and interference.
However, this behavior becomes less evident in the PASS. The low-SNR region where MRT outperforms ZF is significantly narrowed, and ZF only slightly underperforms MRT when $P_d < -15$ dBm. In the remaining power range, ZF nearly always achieves the same sum-rate as MMSE. This is because the adjustment of PAs' positions in PASS enhances the channel gain and enables the system to operate in a regime that matches the characteristics of the chosen beamforming strategy, such as an interference-limited regime for ZF. Moreover, as transmit power increases, the performance gap between MRT and the other two schemes widens, and MRT exhibits a decreasing trend in sum-rate due to its inability to effectively mitigate strong inter-user interference.

Fig. \ref{UL_power} presents the achievable uplink sum-rate versus per-user transmit power under different beamforming schemes. A similar phenomenon is observed as in the downlink case: thanks to the channel reconfigurability of PASS, it consistently outperforms the fixed-antenna system. For instance, at $P_u = 10$ dBm, the sum-rate improvement reaches 10\%. In addition, the low-SNR region where ZF underperforms MRC is also significantly reduced, and the gap between them becomes negligible (e.g., at $P_u = -15$ dBm, the difference is marginal). Likewise, ZF achieves nearly the same performance as MMSE across almost the entire considered power range. MRC suffers from severe inter-user interference at high transmit power, leading to a rise-then-fall trend in sum-rate.

\subsection{Sum-Rate Versus Side Length}
Fig. \ref{fig:DL_UL_length} presents the achievable downlink and uplink sum-rate performances as functions of the region side length $D_x$. In both scenarios, the sum-rate decreases as $D_x$ increases, which is primarily due to the increased average distance between users and the antenna array, resulting in higher large-scale path loss. This performance degradation trend is observed across all beamforming schemes.
Nonetheless, PASS consistently achieves significantly better performance than the conventional fixed-antenna system under all considered configurations. This advantage becomes more pronounced with increasing $D_x$, especially in Fig. \ref{UL_length}, as PASS can dynamically adjust PAs closer to users, thereby effectively compensating for large-scale path loss in larger regions. In contrast, the fixed-antenna system suffers from static deployment and lacks this spatial flexibility. 
\begin{figure*}[!t]
  \centering
  \setlength{\abovecaptionskip}{0pt}
  \subfigure[Achievable downlink sum-rate vs. number of PAs $N$.]{
    \includegraphics[height=0.35\textwidth]{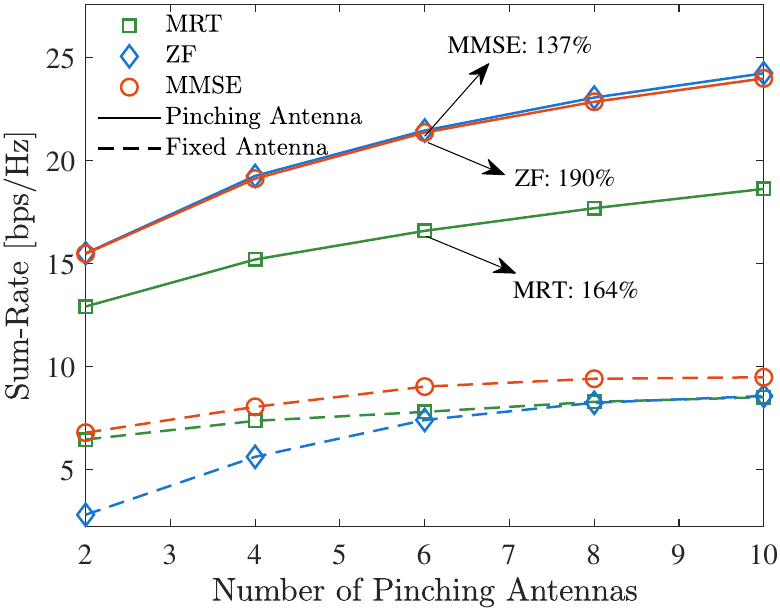}
    \label{DL_PAs}
  }
  \hspace{1cm}
  \subfigure[Achievable uplink sum-rate vs. number of PAs $N$.]{
    \includegraphics[height=0.35\textwidth]{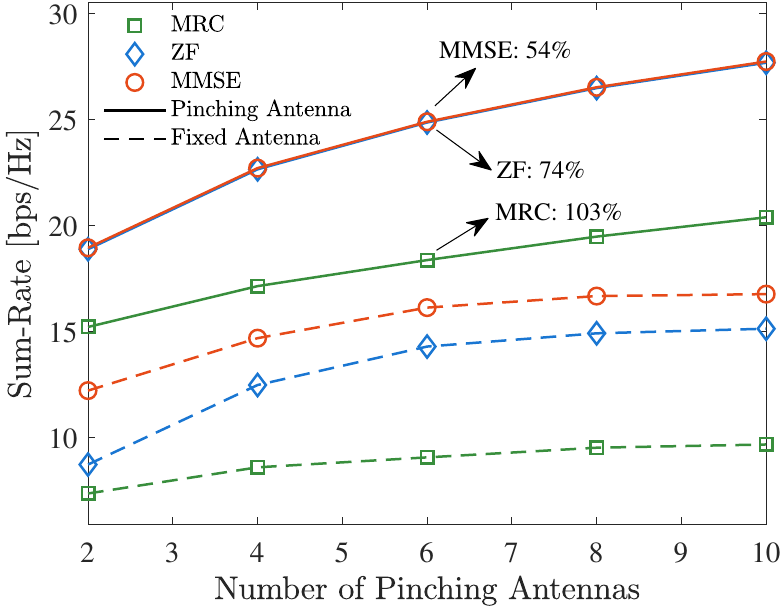}
    \label{UL_PAs}
  }
  \caption{Achievable sum-rate performance under different numbers of PAs.}
  \label{fig:DL_UL_PAs}
\end{figure*}
\begin{figure*}[!t]
  \centering
  \setlength{\abovecaptionskip}{0pt}
  \subfigure[Achievable downlink sum-rate vs. number of users $K$.]{
    \includegraphics[height=0.35\textwidth]{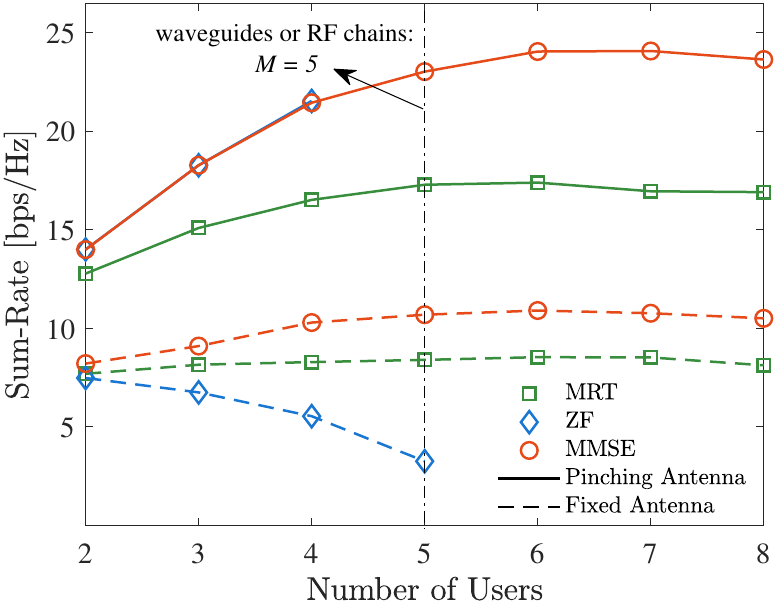}
    \label{DL_users}
  }
  \hspace{1cm}
  \subfigure[Achievable uplink sum-rate vs. number of users $K$.]{
    \includegraphics[height=0.35\textwidth]{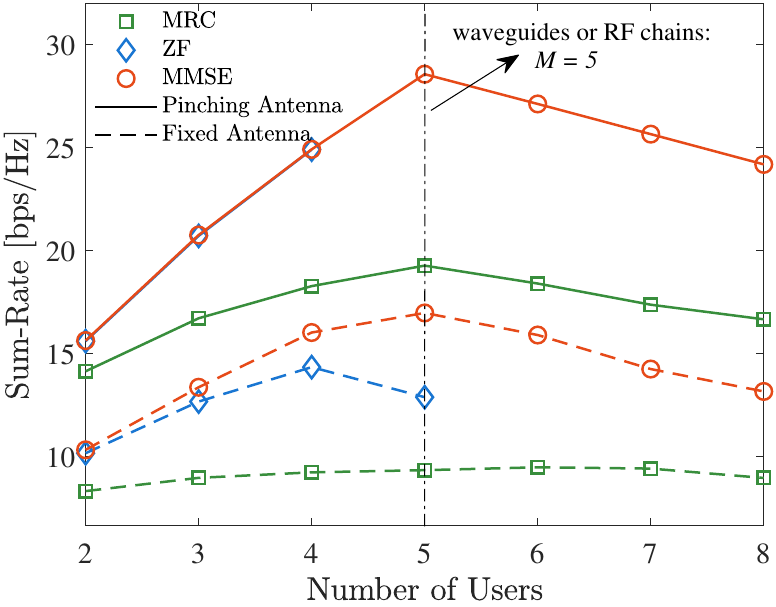}
    \label{UL_users}
  }
  \caption{Achievable sum-rate performance under different numbers of users.}
  \label{fig:DL_UL_users}
\end{figure*}

Moreover, for both downlink and uplink cases, a consistent observation is that in the fixed-antenna system, the performance advantages of MMSE and ZF over MRT (or MRC) diminish rapidly as $D_x$ increases. However, in PASS, such degradation is much slower, and ZF can maintain parity with MMSE across all region sizes. This is because the activation positions of the pinched PAs in PASS can be flexibly adjusted, which allows the system to control the received signal strength and adapt its operating regimes to suit different beamforming strategies. For example, the system can favor an interference-limited regime for ZF or a noise-limited regime for MRT or MRC. As a result, neither ZF nor MMSE suffers significant performance degradation as the average distance increases.

\subsection{Sum-Rate Versus Number of PAs}
Fig. \ref{fig:DL_UL_PAs} shows the achievable sum-rate performance for both downlink and uplink scenarios as a function of the number of PAs $N$. In both cases, the sum-rate improves monotonically with increasing $N$ across all beamforming schemes. This performance gain is primarily attributed to the increased spatial DoFs and array gain provided by the additional antennas.
Across all values of $N$, the proposed PASS architecture consistently outperforms the conventional fixed-antenna system. For instance, when $N=6$, the downlink PASS achieves sum-rate gains of approximately 137\%, 190\%, and 124\% for MMSE, ZF, and MRT, respectively, compared to the fixed-antenna system. In the uplink case, the improvements are also evident, with corresponding gains of 54\%, 74\%, and 103\% for MMSE, ZF, and MRC. 
This is because the reconfigurability of PASS enables flexible adjustment of PAs' positions, which not only facilitates signal phase alignment but also mitigates large-scale path loss. In contrast, fixed-antenna systems rely solely on phase shifters, which can only counteract small-scale fading effects, limiting their overall effectiveness.

\subsection{Sum-Rate Versus Number of users}
Fig. \ref{fig:DL_UL_users} illustrates the downlink and uplink sum-rates achieved by PASS and fixed-antenna systems as a function of the number of users\footnote{It is important to note that when the number of users exceeds the number of waveguides or RF chains, the ZF beamformer becomes inapplicable. In particular, the low-complexity ZF-based antenna activation scheme for PASS requires the number of waveguides $M$ to be greater than the number of users $K$.}. 
As shown in the Fig. \ref{DL_users}, in the downlink scenario, both MRT and MMSE exhibit a ``rise-then-fall'' behavior in PASS and fixed-antenna systems. This trend results from two factors: first, increasing the number of users introduces more severe interference; second, when the number of users $K$ exceeds the number of waveguides or RF chains $M$, the system cannot provide sufficient spatial DoFs to support all data streams. In contrast, ZF behaves differently in the two systems. The fixed-antenna system in this case operates in a noise-limited regime, where interference suppression has limited impact on system sum-rate, and thus the performance naturally degrades as interference increases. In the PASS, however, the flexible adjustment of PAs' positions enables the system to operate in an interference-limited regime, allowing ZF to achieve performance comparable to MMSE.
In the uplink scenario, as shown in Fig. \ref{UL_users}, all three beamforming schemes exhibit similar trends. MMSE and MRC experience noticeable performance degradation when $K>M$, due to insufficient spatial DoFs and increased interference. ZF in PASS maintains performance comparable to MMSE, while in the fixed-antenna system, its sum-rate first increases and then decreases. This is because, as the number of users increases, although ZF can still cancel inter-user interference, it becomes increasingly difficult to align with the users' spatial response, leading to degraded sum-rate.
Overall, although PASS may exhibit similar trends to fixed-antenna systems in certain conditions, the sum-rate improvements it provides remain substantial.

\section{Conclusion}
This paper proposed a low-complexity element-wise optimization scheme for jointly optimizing baseband and pinching beamforming in PASS. By leveraging classical linear beamforming techniques, the original coupled problem was reduced to a single-variable formulation with respect to PAs' positions. For each beamformer, compact closed-form expressions were derived to enable efficient one-dimensional search, avoiding repeated matrix inversions and other costly operations.
Simulation results demonstrated that PASS significantly outperform conventional fixed-antenna systems across various settings. The flexible adjustment of PAs' positions enhanced the channel gain and allowed the system to adapt its operating regime to better suit different beamforming strategies. Notably, the performance gap between ZF and MMSE observed in fixed-antenna systems became negligible in PASS.
This study laid a foundation for the practical deployment of low-complexity joint beamforming solutions for PASS.

\bibliographystyle{IEEEtran}
\bibliography{mybib}
\end{document}